\documentclass[sigconf, screen=true]{acmart}

\renewcommand\footnotetextcopyrightpermission[1]{} %
\usepackage[utf8]{inputenc}
\usepackage{hyphenat}
\usepackage{xspace}
\usepackage{amsmath}
\usepackage{mathtools,amssymb,mathrsfs}
\usepackage{stmaryrd}
\usepackage{amsthm}
\usepackage{stmaryrd}
\usepackage{microtype}
\usepackage{paralist}
\usepackage[boxed]{algorithm2e}
\usepackage{enumerate}
\usepackage{enumitem} 
\usepackage{mdframed}
\usepackage{booktabs} %
\usepackage{url}

\usepackage{multirow}

\usepackage{tikz}
\usepackage{subcaption}
 
\usetikzlibrary{positioning}
\usetikzlibrary{shapes,arrows,calc,matrix}

\newcommand{\nopVLDB}[1]{}

\newcommand*\rot{\rotatebox{90}}

\DeclareMathAlphabet{\mathpzc}{OT1}{pzc}{m}{it}

\newcommand{\cal}{\mathcal}

\newcommand{\calO}{{\mathcal O}}

\newcommand{\ra}{\ensuremath{\rightarrow}}

\newcommand{\ffo}[1]{\ensuremath{\mathit{#1}}}

\newcommand{\tw}{\ffo{tw}}

\newcommand{\icBMIP}{$ic$-BMIP}

\newcommand{\Sp}{\ffo{Sp}}

\newcommand{\fhw}{\ffo{fhw}}
\newcommand{\ghw}{\ffo{ghw}}
\newcommand{\hw}{\ffo{hw}}

\newcommand{\weight}{\ffo{weight}}

\newcommand{\newdetkdecomp}{\texttt{NewDetKDecomp}}
\newcommand{\detkdecomp}{\texttt{DetKDecomp}}
\newcommand{\globalbip}{\texttt{GlobalBIP}}
\newcommand{\localbip}{\texttt{LocalBIP}}
\newcommand{\balsep}{\texttt{BalSep}}
\newcommand{\improvehd}{\texttt{SimpleImproveHD}}

\newcommand{\fracimprovehd}{\texttt{FracImproveHD}}

\newcommand{\iwidth}[1]{\mbox{\it iwidth}(#1)}
\newcommand{\cmiwidth}[2]{\mbox{\it $#1$-miwidth(#2)}}

\newcommand{\VTu}{V(T_u)}

\newcommand{\rec}[1]{\textsc{Check}(#1)}

\newcommand{\classC}{\ensuremath{\mathscr{C}}}

\newcommand{\HH}{\ensuremath{H}}

\newcommand{\np}{{\sf NP}\xspace}

\newif\ifrdfshort
   \rdfshortfalse

\ifrdfshort

\else

\fi

\newcommand{\nop}[1]{}

\newlength{\knotenbreite}
\newlength{\knotenhoehe}
\newlength{\ueberhoehe}
\newlength{\knotenradius}
 \setlist[itemize,1]{leftmargin=\dimexpr 26pt-.15in}

\newenvironment{myintro}%
  {\list{}{\leftmargin=0.1in\rightmargin=0.1in}\item[]}%
  {\endlist}

\title{HyperBench: A Benchmark and Tool for Hypergraphs \\
and Empirical Findings}

\author{Wolfgang Fischl$^1$, Georg Gottlob$^{1,3}$, Davide M. Longo$^{1,2}$, and Reinhard Pichler$^1$}
\affiliation{%
  \institution{$^1$TU Wien, $^2$Universit\`{a} della Calabria, and $^3$University of Oxford}
}

\begin{document}

\begin{abstract}
To cope with the intractability of answering Conjunctive 
Queries (CQs) and solving Constraint Satisfaction Problems (CSPs), several notions of hypergraph decompositions have been proposed -- giving rise to different notions of width, noticeably, plain, generalized, and fractional hypertree width (hw, ghw, and fhw). Given the increasing interest in using such decomposition methods in practice, a publicly accessible repository of decomposition software, as well as a large set of benchmarks, and a web-accessible workbench for inserting, analysing, and retrieving hypergraphs are called for. 

We address this need by providing (i) concrete implementations of hypergraph decompositions (including new practical algorithms), (ii) a new, comprehensive benchmark of hypergraphs stemming from disparate CQ and CSP collections, and (iii) HyperBench, our new web-inter\-face for accessing the benchmark and the results of our analyses. In addition, we describe a number of actual experiments we carried out with this new infrastructure.
\end{abstract}

\maketitle

\nocite{DBLP:journals/siamcomp/AtseriasGM13,malyshevgetting}

\section{Introduction}
\label{sect:introduction}

In this work we study computational problems on hypergraph decompositions 
which are designed to speed up 
the evaluation of Conjunctive Queries (CQs) 
and the solution of Constraint Satisfaction Problems (CSPs).
Hypergraph decompositions have 
meanwhile found their way into commercial database systems such as LogicBlox 
\cite{DBLP:conf/sigmod/ArefCGKOPVW15,
olteanu2015size,BKOZ13,KhamisNRR15,KhamisNR16} and advanced research prototypes 
such as 
EmptyHeaded~\cite{DBLP:journals/tods/AbergerLTNOR17,aberger2016old,tu2015duncecap,Duncecap15a}.  
Hypergraph decompositions have also been successfully used in 
the CSP area \cite{DBLP:journals/aicom/AmrounHA16,DBLP:journals/jetai/HabbasAS15,LalouHA09}.
In theory, 
the pros and cons of various notions of decompositions and widths are well understood 
(see \cite{DBLP:conf/pods/GottlobGLS16} for a survey).
However, from a practical point of view, many questions have remained~open. 

We want to analyse the hypertree width ($\hw$) of hypergraphs  from different application contexts. 
The investigation of millions of CQs \cite{DBLP:journals/pvldb/BonifatiMT17,DBLP:conf/sigmod/PicalausaV11}
posed at various SPARQL endpoints suggests that these real-world CQs with atoms of arity $\leq 3$ 
have very low $\hw$: the overwhelming majority is acyclic; almost all of the rest
has $\hw = 2$. It is, however, not clear if CQs with arbitrary arity and CSPs also have 
low hypertree width,  
say, $\hw \leq 5$. Ghionna et al.\ \cite{DBLP:conf/icde/GhionnaGGS07} gave a positive answer to this question 
for a small set of TPC-H  
benchmark \linebreak
queries. We significantly extend their collection of CQs.

Answering CQs and solving CSPs are fundamental tasks in Computer Science. Formally, they are the same problem, since both  correspond to the evaluation of first-order formulae over a finite structure, such that the formulae only use $\{\exists,\wedge\}$ 
as connectives
but not $\{\forall, \vee, \neg\}$.
Both problems, answering CQs and solving CSPs, 
are \np-complete \cite{DBLP:conf/stoc/ChandraM77}. 
Consequently, the search for tractable fragments of these problems has been an active research area in the database and artificial intelligence communities for several decades. 

The most powerful methods 
known to date 
for defining tractable fragments are based on various decompositions of the 
hypergraph structure 
underlying 
a given CQ or CSP. 
\nopVLDB{*****************************
Here, ``tractability'' means solvability in polynomial time (for boolean CQs and CSPs) or the enumeration of all answers (to a given CQ) or of all solutions (to a given CSP) with polynomial delay.
*****************************}
The most important forms of decompositions 
are 
{\it hypertree decompositions (HDs)\/}  
\cite{2002gottlob}, 
{\it generalized 
hypertree decompositions (GHDs)}~\cite{2002gottlob}, and 
{\it fractional hypertree decompositions (FHDs)}~\cite{2014grohemarx}.
These decomposition methods give rise to three
notions of width of a hypergraph $H$:
the  {\it hypertree width} $\hw(H)$,  
{\em generalized hypertree width} $\ghw(H)$, 
and {\em fractional hypertree width} 
$\fhw(H)$, 
where, $\fhw(H)\leq \ghw(H)\leq \hw(H)$ holds for every 
hypergraph $H$. 
For definitions, see Section~\ref{sect:prelim}.

Both, answering CQs and solving CSPs,
become tractable if 
the underlying hypergraphs have bounded 
$\hw$, $\ghw$, or, $\fhw$ and an appropriate decomposition is given. 
This gives rise to the 
problem of recognizing if a given CQ or CSP has $\hw$, $\ghw$, or, $\fhw$ bounded by 
some constant $k$. 
Formally, for 
{\it decomposition\/} $\in \{$HD, GHD, FHD$\}$ and $k \geq 1$, 
we consider the following family~of~problems:

\smallskip 
\noindent
\rec{{\it decomposition\/},\,$k$}\\
\begin{tabular}{ll}
 \bf input & hypergraph $H = (V,E)$;\\
 \bf output & {\it decomposition\/}  of $H$ of width $\leq k$ if it 
exists and  \\ & answer `no' otherwise.
\end{tabular}
\smallskip 

Clearly, bounded $\fhw$ defines the 
largest tractable class 
while bounded 
$\hw$ defines the smallest one. 
On the other hand, the problem
\rec{HD,\,$k$}\  is feasible in polynomial time
\cite{2002gottlob}
while
the problems \rec{GHD,\,$k$} \cite{2009gottlob} 
and 
\rec{FHD,\,$k$} \cite{pods/FischlGP18} 
are \np-complete even for $k = 2$.

Systems to solve 
the \rec{HD,\,$k$} problem 
exist \cite{DBLP:journals/jea/GottlobS08,DBLP:journals/jcss/ScarcelloGL07}. 
In contrast, 
for the problems \rec{GHD,\,$k$} 
and 
\rec{FHD,\,$k$}, 
apart from exhaustive search over possible decomposition trees 
\linebreak
(which only works for small hypergraphs),  
no implementations have been reported yet \cite{DBLP:journals/tods/AbergerLTNOR17} -- with one exception: 
very recently, an interesting approach is presented in \cite{CP/FichteHLS18},
where SMT-solving is 
\linebreak
applied to the \rec{FHD,\,$k$} problem.
In \cite{DBLP:journals/jea/GottlobS08}, 
tests of the \linebreak 
\rec{HD,\,$k$} system 
are presented. However, 
a benchmark for systematically evaluating systems 
for the 
\rec{{\it decomposition\/},\,$k$} problem
with {\it decomposition\/} $\in \{$HD, GHD, FHD$\}$ and $k \geq 1$
were missing so far.
This motivates our first research goals.

\begin{myintro}
\noindent{\bf Goal 1:} Create a comprehensive,  easily extensible benchmark of hypergraphs 
corresponding to CQs or CSPs for the analysis of hypergraph decomposition algorithms.
\end{myintro}

\begin{myintro}
\noindent{\bf Goal 2:} Use the %
benchmark from Goal 1 to find out if the hypertree width is, in general, small enough (say $\leq 5$) to allow for efficient evaluation of CQs of arbitrary arity and of CSPs. 
\end{myintro} 

Recently, in  \cite{pods/FischlGP18}, 
the authors have identified classes of CQs for which the 
\rec{GHD,\,$k$}  and \rec{FHD,\,$k$} problems
become tractable
(from now on, we only speak about CQs; of course, all results apply equally to CSPs). To this end, the Bounded 
Intersection Property (BIP) and, more generally, the 
Bounded Multi-Intersection Property (BMIP)
have been introduced. 
The maximum number $i$ of attributes shared by two 
(resp.\ $c$) atoms is referred to as the 
intersection width (resp.\ $c$-multi-intersection width) of the CQ, which 
is similar to the notion of cutset width from the CSP literature \cite{dechter2003}. 
We say that a class of CQs satisfies the BIP (resp.\ BMIP) 
if the number of attributes shared by two
(resp.\ by a constant number $c$ of) query atoms is bounded by 
some constant $i$. 

A related property is that of bounded degree, i.e., 
each attribute only occurs in a constant number of query atoms. Clearly, the BMIP
is an immediate consequence of bounded degree.
It has been shown in \cite{pods/FischlGP18}
that \rec{GHD,\,$k$} is solvable in polynomial time for 
CQs whose underlying hypergraphs satisfy the BMIP. 
For CQs, the BMIP and bounded degree seem natural restrictions. For CSPs, the situation is not so clear. 
This yields the following research goals.

\begin{myintro}
\noindent{\bf Goal 3:} Use the hypergraph benchmark from Goal~1 to 
analyse how realistic the restrictions to low (multi-)inter\-section width,
or low degree of CQs and CSPs are.
\end{myintro} 

\begin{myintro}
\noindent{\bf Goal 4:} Verify 
that the tractable fragment of the \rec{GHD,\,$k$} problem 
given by hypergraphs of low intersection width
indeed allows for efficient algorithms that work well in practice.
\end{myintro}

The tractability results for \rec{FHD,\,$k$} \cite{pods/FischlGP18,FGP2017tractablefhw}
are significantly weaker than for 
\rec{GHD,\,$k$}: they involve a factor which is at least double-exponential 
in some ``constant'' (namely $k$, the bound $d$ on the degree and/or the bound $i$ on the intersection-width). Hence, we want to investigate if (generalized) hypertree decompositions 
could be ``fractionally improved'' 
by taking the integral edge cover at each node in the 
HD or GHD and replacing it by a fractional edge cover. 
We will thus introduce the notion of {\em fractionally improved\/} HD 
which checks if there exists an HD of width $\leq k$, 
such that replacing each integral cover by a fractional cover yields an FHD of width
$\leq k'$ for given bounds $k,k'$ with $0 < k' < k$.

\begin{myintro}
\noindent{\bf Goal 5:} Explore the potential of fractionally improved HDs, i.e., 
investigate if the improvements %
achieved are significant.
\end{myintro}

In cases where \rec{GHD,\,$k$} and 
\rec{FHD,\,$k$} are intractable, we may have to settle for 
good approximations of $\ghw$ and $\fhw$. 
For GHDs, we may thus use the 
inequality  $\ghw(H) \leq 3 \cdot \hw(H) +1$, which 
holds for every 
hypergraph $H$ \cite{DBLP:journals/ejc/AdlerGG07}.
In contrast, for FHDs, the best known general, polynomial-time approximation is cubic. More precisely, 
in \cite{DBLP:journals/talg/Marx10}, a polynomial-time algorithm is presented which, 
given a hypergraph $H$ with $\fhw(H) = k$,
computes an FHD of width $\calO(k^3)$. 
In \cite{pods/FischlGP18}, it is shown that a polynomial-time approximation up to a logarithmic factor is 
possible for any class of hypergraphs with bounded 
Vapnik--Chervonenkis dimension (VC-dimension; see
Section~\ref{sect:prelim} for a precise definition).
The problem of efficiently approximating the $\ghw$ 
and/or $\fhw$ 
leads us to the following goals.

\begin{myintro}
\noindent{\bf Goal 6:} Use the %
benchmark from Goal~1 to 
analyse if, in practice, $\hw$ and $\ghw$ indeed differ by factor 3 or, if $\hw$ is typically much closer to $\ghw$
than this worst-case bound.
\nopVLDB{than this worst-case bound.}
\end{myintro}

\begin{myintro}
\noindent{\bf Goal 7:} Use the %
benchmark from Goal~1 to 
analyse how realistic the restriction to small VC-dimen\-sion of 
CQs and CSPs is.
\end{myintro} 

\smallskip
\noindent
{\bf Results.}
By pursuing these goals, we obtain the following results:

\smallskip
$\bullet$ \
We provide {\em HyperBench\/}, a comprehensive hypergraph benchmark of %
initially over %
3,000 hypergraphs (%
see Section~\ref{sect:HyperBench}). 
This benchmark is exposed by a web interface, 
which allows the user to retrieve the hypergraphs or groups of hypergraphs together with a broad spectrum of properties of these hypergraphs, such as 
lower/upper bounds on $\hw$ and $\ghw$, (multi-)intersection width, degree, 
etc.

\smallskip

$\bullet$  \
We extend the software for HD computation from \cite{DBLP:journals/jea/GottlobS08}
to also solve the 
\rec{GHD,\,$k$} problem. 
For a given hypergraph $H$, our system first computes the intersection width of $H$ and then applies the 
$\ghw$-algo\-rithm from \cite{pods/FischlGP18},
which is parameterized by the
intersection width. We implement several improvements and 
we further extend the system to compute
also  ``fractionally improved''~HDs.

\smallskip

$\bullet$  \
We carry out an empirical analysis of the hypergraphs in the HyperBench benchmark. This analysis demonstrates, especially for real-world instances, 
that the restrictions to BIP, BMIP, bounded degree, and bounded VC-dimension are astonishingly realistic. Moreover,
on all hypergraphs in the HyperBench benchmark, we 
run our $\hw$- and $\ghw$-systems to identify 
(or at least bound) their $\hw$ and $\ghw$.
An interesting observation of our empirical study is that apart from the CQs 
also a significant portion of CSPs in ourbenchmark 
has small hypertree width (all non-random CQs have $\hw \leq 3$ and 
over 60\% of CSPs stemming from applications have 
$\hw \leq 5$). Moreover, for $\hw \leq 5$, 
in all of the 
cases where the $\ghw$-computation terminates, $\hw$ and $\ghw$ 
have identical values.

\smallskip

$\bullet$  \
In our study of the $\ghw$ of the hypergraphs in the HyperBench benchmark, 
we observed that a straightforward implementation of the algorithm from \cite{pods/FischlGP18} 
for hypergraphs of 
low intersection width is too slow in many cases. We therefore present a new approach (based on so-called
``balanced separators'') with promising experimental results. It is interesting to note that
the new approach 
works particularly well in those situations which are particularly hard for the straightforward implementation,
namely hypergraphs $H$ where the test if $\ghw \leq k$ for given $k$ gives a ``no''-answer. 
Hence, combining the different approaches is very effective.

\smallskip

\noindent 
{\bf Structure.} This paper is structured as follows:
In Section~\ref{sect:prelim}, we recall some basic definitions and results.  
In Section~\ref{sect:HyperBench}, we present our system and test environment as well as our hypergraph benchmark 
HyperBench. 
First results of our empirical study of the hypergraphs in 
this benchmark are presented in 
Section~\ref{sect:EmpiricalAnalysis}.
In Section~\ref{sect:ghw-implementation}, we describe our algorithms for solving the 
\rec{GHD,\,$k$} problem.
A further extension of the system to allow for the computation of fractionally improved HDs 
is described in Section~\ref{sect:fractionally-improved}.
Finally, in Section~\ref{sect:related} we summarize related work and
conclude in Section~\ref{sect:conclusion} by highlighting the most important lessons learned  from our empirical study and 
by identifying some appealing directions for future~work. 

\smallskip
\noindent
Due to lack of space, some of the statistics presented in the main body contain aggregated values (for instance, for different classes of CSPs). Figures and tables with more 
fine-grained results (for instance, distinguishing the 3 classes of CSPs to be presented
in Section~\ref{sect:EmpiricalAnalysis}) are provided in the appendix and will be made
publically  available in a full version of this paper in CoRR.

\section{Preliminaries}
\label{sect:prelim}

Let $\phi$ be a CQ or CSP (i.e., an FO-formula with connectives $\{\exists,\wedge\}$).
The {\em hypergraph corresponding to $\phi$} 
is defined as 
$H=(V(H),E(H))$, where the set of vertices $V(H)$ is defined as the set of variables in $\phi$ 
and the set of edges $E(H)$ is defined as $E(H) = \{ e \mid \phi$ contains an atom $A$, s.t.\
$e$ equals the set of variables occurring in $A\}$.

\smallskip

\noindent
{\bf Hypergraph decompositions and width measures.}
We consider here three notions of 
hypergraph decompositions with associated notions of width. To this end, we first need to introduce 
the notion of (fractional) edge covers:

Let $H = (V(H),E(H))$ be a hypergraph 
and consider a function
$\gamma \colon E(H) \ra [0,1]$. Then, we define  the set $B(\gamma)$ of all 
vertices covered by $\gamma$ and the weight of $\gamma$ as 
\begin{eqnarray*}
B(\gamma) & = &\left\{ v\in V(H) \mid \sum_{e\in E(H), v\in e} \gamma(e) \geq 1 
\right\}, \\
\ \weight(\gamma) & =& \sum_{e \in E(H)} \gamma(e).
\end{eqnarray*}
The special case of a function with values restricted to $\{0,1\}$, will 
be denoted by $\lambda$, i.e.,
$\lambda \colon E(H) \ra \{0,1\}$. 
Following \cite{2002gottlob}, we can also treat %
$\lambda$ as a set with $\lambda \subseteq E(H)$ 
(namely, the set of edges $e$ with $\lambda(e) = 1$)
and the weight as the cardinality of such a set of edges. 

We now introduce three notions of hypergraph decompositions.

\begin{definition}
 \label{def:GHD}
A {\em generalized hypertree decomposition\/} 
(GHD) of a hypergraph 
$H=(V(H),E(H))$ 
is a tuple 
$\left< T, (B_u)_{u\in N(T)},\right.$ \linebreak $\left.(\lambda_u)_{u\in N(T)} \right>$, such that 
$T = \left< N(T),E(T)\right>$ is a rooted tree and 
the 
following conditions hold:

\begin{enumerate}[label=\emph{(\arabic*})]
\setlength{\itemsep}{0.01ex}
 \item %
 $\forall e \in E(H)$: there exists a node $u \in N(T)$ with $e \subseteq 
B_u$;
 \item %
 $\forall v \in V(H)$: the set $\{u \in N(T) \mid v \in B_u\}$ is 
connected 
in $T$;
 \item %
 $\forall u\in N(T)$: $\lambda_u$ is defined as $\lambda_u\colon E(H) 
\ra 
\{0,1\}$
with 
$B_u \subseteq  B(\lambda_u)$.
\end{enumerate}
\end{definition}

We use the following notational conventions throughout this paper.
To avoid confusion, we will consequently refer to the 
elements in 
$V(H)$ as {\em vertices\/} of the hypergraph and to the 
elements in $N(T)$ as 
the {\em nodes\/} of the decomposition. 
For a node $u$ in $T$, 
we write $T_u$ to denote the subtree of $T$ rooted at $u$.
By slight abuse of notation, we will often write $u' \in T_u$ to denote
that $u'$ is a node in the subtree $T_u$ of $T$.
Finally, we define $\VTu 
:= \bigcup_{u' \in T_u} B_{u'}$.

\begin{definition}
 \label{def:HD}
 A {\em hypertree decomposition\/} (HD) of a hypergraph 
$H=(V(H),E(H))$  is a GHD, which in addition also 
satisfies the following condition:
\begin{enumerate}[label=\emph{(\arabic*})]
 \item[(4)] %
 $\forall u\in N(T)$: $ V(T_u) \cap B(\lambda_u) \subseteq B_u$ 
\end{enumerate}
\end{definition}

\begin{definition}
 \label{def:FHD}
 A 
{\em fractional hypertree decomposition\/} 
(FHD) 
\cite{2014grohemarx}
of a hypergraph 
$H=(V(H),E(H))$ is a tuple 
$\left< T, (B_u)_{u\in N(T)},\right.$ \linebreak $\left. (\gamma_u)_{u\in N(T)} \right>$, where
conditions (1) and (2) of Definition~\ref{def:GHD} plus the following condition 
(3') hold:
\begin{enumerate}[label=\emph{(\arabic*})]
 \item[(3')] $\forall u\in N(T)$: 
$\gamma_u$ is defined as 
$\gamma_u\colon E(H) 
\ra 
[0,1]$
with 
$B_u \subseteq  B(\gamma_u)$.
\end{enumerate}
\end{definition}

The width of a GHD, HD, or FHD is the maximum weight of the functions 
$\lambda_u$ or $\gamma_u$, 
over all nodes $u$ in $T$. 
The generalized hypertree 
width,
hypertree width, and fractional hypertree width of $H$ (denoted $\ghw(H)$, 
$\hw(H)$, $\fhw(H)$) is the minimum width over all GHDs, HDs, and FHDs of $H$, 
respectively.
Condition~(2) is called the ``connectedness condition'', and condition~(4) is 
referred to as ``special condition'' \cite{2002gottlob}. The set $B_u$ is often 
referred to as the 
``bag'' at node $u$. The functions $\lambda_u$ and $\gamma_u$
are referred to as the $\lambda$-{\em label\/} and $\gamma$-{\em label\/} of node $u$.
Strictly speaking, only HDs require that the underlying tree $T$ be rooted. 
We assume that also the tree underlying a 
GHD or an FHD is rooted where the root is {\em arbitrarily\/} 
chosen.

\medskip
\noindent
{\bf Favourable properties of hypergraphs.}
In \cite{pods/FischlGP18}, the following properties of 
hypergraphs were identified to allow for the definition of tractable classes of 
\rec{GHD,\,$k$} and  for an efficient approximation of \rec{FHD,\,$k$}, respectively.

\begin{definition}\label{def:bip}
The {\em intersection width} $\iwidth{\HH}$ of a hypergraph $\HH$ is the 
maximum 
cardinality of any intersection $e_1\cap e_2$ of two edges $e_1 \neq e_2$ 
of $\HH$.
We say that a hypergraph $H$ has
the {\em $i$-bounded intersection property ($i$-BIP)} if 
$\iwidth\HH\leq i$. 
A class $\classC$  of hypergraphs
has the  
{\em bounded intersection property (BIP)} if there exists some 
constant $i$ such that
every hypergraph $H$ in $\classC$
 has the $i$-BIP.
\end{definition}

\begin{definition}\label{def:bmip}
For positive integer $c$, 
the {\em $c$-multi-inter\-section width} $\cmiwidth{c}{\HH}$ of a 
hypergraph 
$\HH$ 
is the maximum cardinality of any intersection $e_1\cap\cdots\cap e_c$  of 
$c$ distinct edges $e_1, \ldots,  e_c$ of $\HH$.  
We say that a hypergraph $H$ has
the 
{\em $i$-bounded $c$-multi-intersection property (\icBMIP)} if 
$\cmiwidth{c}{\HH}\leq i$ holds. 
We say that 
a class $\classC$ of hypergraphs
has the  
{\em bounded multi-intersection property (BMIP)} 
if there exist constants $c$ and $i$ 
such that
every hypergraph $H$ in $\classC$
has the \icBMIP.
\end{definition}

There are two more relevant properties of (classes of) hypergraphs: bounded degree and bounded
Vapnik--Chervo\-nenkis dimension (VC-dimension). 
It is easy to verify \cite{pods/FischlGP18} that bounded degree implies the BMIP,
which in turn implies bounded VC-dimension.

\begin{definition}\label{def:bounded-degree}
The {\em degree} $\deg(\HH)$ of a 
hypergraph 
$\HH$ 
is defined as the 
maximum number $d$ of hyperedges in which a vertex occurs, i.e., 
$d = \max_{v \in V(H)} |\{ e \in E(H) \mid v \in E(H)\}|$.
We say that a class $\classC$ of hypergraphs  has bounded degree, if there exists $d \geq 1$, 
such that every hypergraph $H\in {\classC}$ has degree $\leq d$.
\end{definition}

\begin{definition}[\cite{1971vc}]
\label{def:vc}
Let $\HH=(V(H),E(H))$ be a hypergraph, and $X\subseteq V$ a set of vertices. 
Denote by 
$E(H)|_X =\{X \cap e\, |\, e\in E(H)\}$. $X$ is called {\em shattered} if 
$E(H)|_X=2^X$.
The {\em Vapnik-Chervonenkis dimension (VC dimension)} of $\HH$ is 
the maximum cardinality of a shattered subset of $V$. 
We say that a class $\classC$ of hypergraphs  has bounded VC-dimension, if there exists $v \geq 1$, 
such that every hypergraph $H\in {\classC}$ has VC-dimension $\leq v$.
\end{definition}

The above four properties help to solve or approximate the
 \linebreak 
\rec{GHD,\,$k$} and \rec{FHD,\,$k$} problems as follows:

\begin{theorem}[\cite{FGP2017tractablefhw,pods/FischlGP18}]
\label{theo:tractability}
Let $\classC$ be a class of hypergraphs. 
\begin{itemize}
\setlength{\itemsep}{0.01ex}

\item If $\classC$ has the BMIP, then the \rec{GHD,\,$k$} problem is solvable in polynomial time for arbitrary $k \geq 1$. 
Consequently, this tractability holds if $\classC$ has bounded degree or the BIP (which each imply the BMIP)
\nopVLDB{\cite{pods/FischlGP18}}
\cite{pods/FischlGP18}.

\item If $\classC$ has bounded degree, then the \rec{FHD,\,$k$} problem is solvable in polynomial time for arbitrary 
$k \geq 1$
\cite{FGP2017tractablefhw}. 
\nopVLDB{\cite{FGP2017tractablefhw}.} 

\item If $\classC$ has bounded VC-dimension, then the $\fhw$ can be approximated in polynomial time up to a 
logarithmic factor \cite{pods/FischlGP18}.
\nopVLDB{***************
i.e., there exists a polynomial-time algorithm that, given a hypergraph $H \in 
\classC$ with  
$\fhw(H) \leq k$, finds an FHD of $H$ of width $\calO(k \cdot \log k)$
\cite{pods/FischlGP18}.
***************}
\end{itemize}

\end{theorem}

\section{HyperBench benchmark and tool}
\label{sect:HyperBench}

In this section, we introduce our system, test environment, and HyperBench -- our new hypergraph benchmark and web tool.

\medskip
\noindent
{\bf System and Test Environment.}
In \cite{DBLP:journals/jea/GottlobS08}, an implementation 
(called \detkdecomp)
of the hypertree decomposition algorithm 
from \cite{2002gottlob} was presented. 
We have extended this 
implementation and built our
new library (called \newdetkdecomp) upon it.
This  library includes the original $\hw$-algorithm from \cite{DBLP:journals/jea/GottlobS08}, the tool {\tt hg-stats} to determine properties described in Section~\ref{sect:EmpiricalAnalysis} and the algorithms to be presented in 
Sections~\ref{sect:ghw-implementation} and~\ref{sect:fractionally-improved}. 
The library is written in C++ and comprises around 8,500 lines of code.
The code is %
available in GitHub at \url{http://github.com/TUfischl/newdetkdecomp}.

\nop{****************
To solve the linear programs (LPs) for computing
fractional covers in our algorithms in Section~\ref{sect:fractionally-improved}, we use the \textit{COIN-OR Linear Programming Solver} (CLP) version~1.16 from \url{https://projects.coin-or.org/Clp}. We have chosen CLP over other open-source LP solvers based on very promising empirical results reported in \cite{gearhart2013comparison}.
****************}

All the experiments reported in this paper were performed on a cluster of 10 workstations each running Ubuntu 16.04. Every workstation has the same specification and is equip\-ped with two Intel Xeon E5-2650 (v4) processors each having 12 cores and 256-GB main memory. Since all algorithms are single-threaded, we were allowed to compute several instances in parallel. For all upcoming runs of our algorithms we set a timeout of 3600s.

\begin{figure}[htbp]
\centering 
     \fbox{\includegraphics[width=0.45\textwidth]{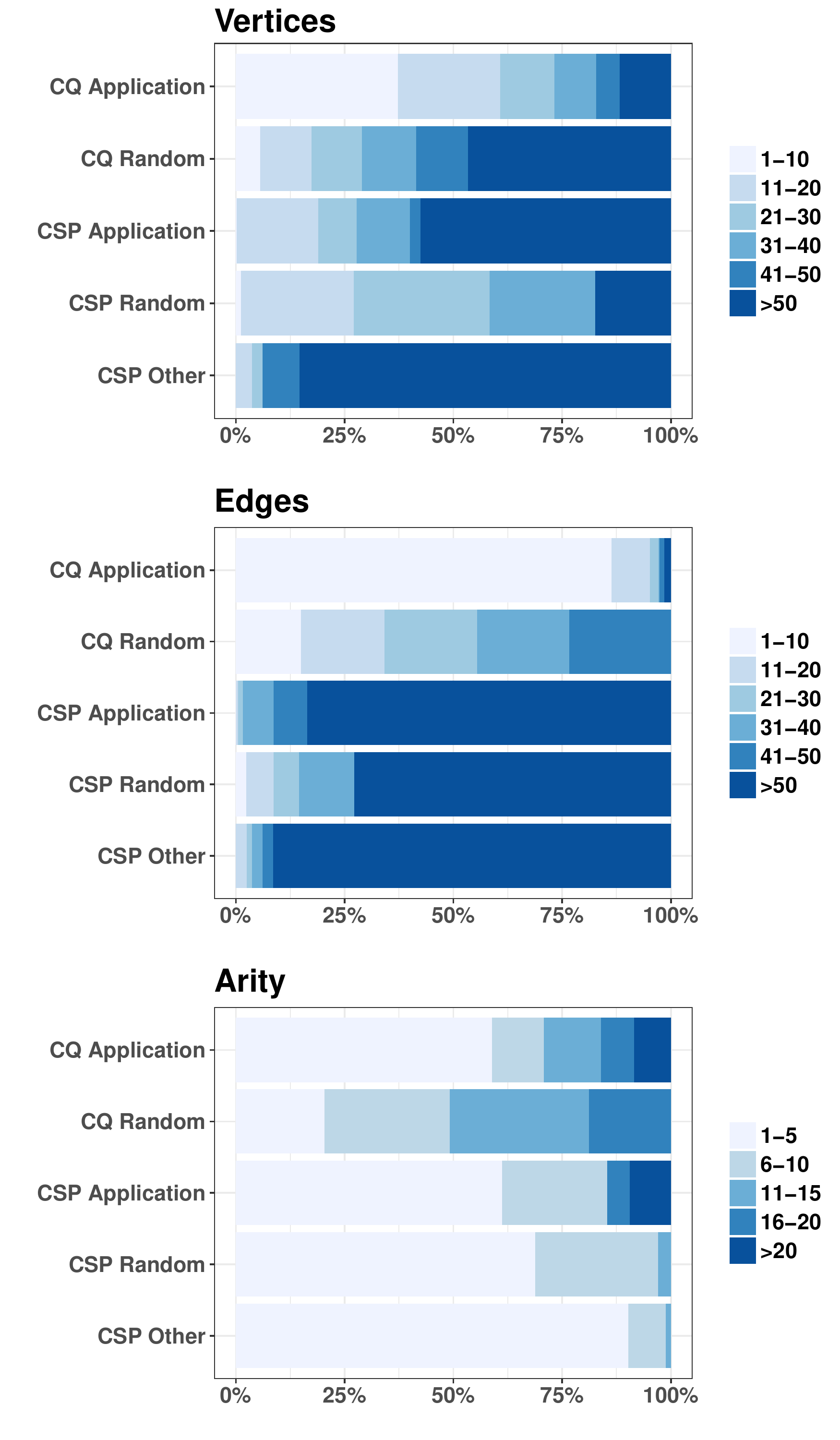}}
    \caption{Hypergraph Sizes} 
\label{fig:hg-sizes}         
\end{figure}    

\begin{table}
  \centering
  \caption{Overview of benchmark instances}
    \begin{tabular}{l|l|r|r}
    & \multicolumn{1}{c|}{\textit{Benchmark}} & \multicolumn{1}{c|}{\textit{No. instances}} & \multicolumn{1}{c}{$\hw\geq 2$} \\
    \hline
    \parbox[t]{2mm}{\multirow{8}{*}{\rotatebox[origin=c]{90}{CQs}}}& \textsc{SPARQL}\cite{DBLP:journals/pvldb/BonifatiMT17} & 70 (out of 26,157,880) & 70 \\
    & \textsc{LUBM}\cite{DBLP:conf/pods/BenediktKMMPST17,guo2005lubm}  & 14    & 2 \\
    & \textsc{iBench}\cite{DBLP:conf/pods/BenediktKMMPST17,arocena2015ibench} & 40    & 0 \\
    & \textsc{Doctors}\cite{DBLP:conf/pods/BenediktKMMPST17,geerts2014mapping} & 14    & 0 \\
    & \textsc{Deep}\cite{DBLP:conf/pods/BenediktKMMPST17}  & 41    & 0 \\
    & \textsc{JOB} (IMDB) \cite{Leis2017} & 33    & 7 \\
    & \textsc{TPC-H}  \cite{BenediktCQs,tpch} & 33    & 1 \\
    & \textsc{SQLShare} \cite{shrjainSQLShare} & 290 (out of 15,170) & 1 \\
    & \textsc{Random} \cite{2001Pottinger} & 500 & 464 \\
    \hline
    \parbox[t]{2mm}{\multirow{3}{*}{\rotatebox[origin=c]{90}{CSPs}}} & \textsc{Application} \cite{xcsp} & 1,090  & 1,090 \\
    & \textsc{Random} \cite{xcsp} & 863   & 863 \\
    & \textsc{Other} \cite{DBLP:journals/jea/GottlobS08,berg2017maxsat} & 82    & 82 \\ 
    \hline
    \multicolumn{1}{r}{} & \multicolumn{1}{r|}{\textit{Total:}} & \textit{3,070} & \textit{2,580} \\ 
    \end{tabular}%
  \label{tab:bench-overview}%
\end{table}%

\medskip
\noindent
{\bf Hypergraph benchmark.}
Our benchmark 
contains 3,070 hypergraphs, which have been converted from CQs and CSPs collected from various sources.
Out of these 3,070 hypergraphs, 2,918 hypergraphs have never been used in a hypertree width analysis before. The hypertree width of 
70 CQs and of 82 CSPs has been analysed in \cite{DBLP:journals/jea/GottlobS08}, \cite{berg2017maxsat}, and/or \cite{DBLP:journals/pvldb/BonifatiMT17}.
An overview of all instances of CQs and CSPs is given in Table~\ref{tab:bench-overview}. 
They have been collected from various publically available benchmarks and repositories
of CQs and CSPs. In the first column, the names of each collection of CQs and CSPs are given together with 
references where they were first published. In the second column we display the number of hypergraphs
extracted from each collection. The $\hw$ of the CQs and CSPs in our benchmark will be discussed
in detail in Section \ref{sect:EmpiricalAnalysis}. To get a first feeling of the $\hw$ of the various sources, 
we mention the number of cyclic hypergraphs (i.e., those with $\hw\geq 2$) in the last column.
When gathering the 
CQs, we proceeded as follows: of the huge benchmark reported in \cite{DBLP:journals/pvldb/BonifatiMT17}, 
we have only included CQs, which were detected as having $\hw \geq 2$ in \cite{DBLP:journals/pvldb/BonifatiMT17}. Of the big repository reported in 
\cite{shrjainSQLShare}, we have included those CQs, which are not trivially acyclic (i.e., they have at least 3 atoms). Of all the small collections of queries, we have included all. 

Below, we describe the different benchmarks in detail:

\smallskip
$\bullet$ \
{\it CQs}: Our benchmark contains 535 CQs from four main sources \cite{BenediktCQs,DBLP:conf/pods/BenediktKMMPST17,DBLP:journals/pvldb/BonifatiMT17,shrjainSQLShare} and a set of 500 randomly generated queries using the query generator of \cite{2001Pottinger}. 
In the sequel, we shall refer to the former queries as {\it CQ Application\/}, and to the latter
as {\it CQ Random\/}.
The CQs analysed in \cite{DBLP:journals/pvldb/BonifatiMT17} constitute by far the biggest repository of CQs -- namely 26,157,880 CQs stemming from 
SPARQL queries. The queries come from real-users of SPARQL endpoints and their hypertree width was already determined in \cite{DBLP:journals/pvldb/BonifatiMT17}. Almost all of these CQs were shown to be acyclic. Our analysis comprises 70 CQs from \cite{DBLP:journals/pvldb/BonifatiMT17}, 
which (apart from few exceptions) are essentially the ones in \cite{DBLP:journals/pvldb/BonifatiMT17} with $\hw \geq 2$. 
In particular, we have analysed all 8 CQs with highest $\hw$ among the CQs analysed in \cite{DBLP:journals/pvldb/BonifatiMT17} 
(namely, $\hw = 3$).

The \textsc{LUBM} \cite{guo2005lubm}, \textsc{iBench} \cite{arocena2015ibench}, \textsc{Doctors} \cite{geerts2014mapping}, and \textsc{Deep} scenarios have been recently used to evaluate the performance of chase-based systems \cite{DBLP:conf/pods/BenediktKMMPST17}. Their queries were especially tailored towards the evaluation of query answering tasks of such systems. 
Note that the \textsc{LUBM} benchmark \cite{guo2005lubm} is a widely used standard benchmark for the evaluation of Semantic Web repositories. Its queries are designed to measure the performance of those repositories over large datasets. Strictly speaking, the \textsc{iBench} is a tool for generating schemas, constraints, and mappings for data integration tasks. However, in \cite{DBLP:conf/pods/BenediktKMMPST17}, 40 queries were created for tests with the 
\textsc{iBench}. We therefore refer to these queries as \textsc{iBench}-CQs here. In summary, we have incorporated all queries that were either contained in the original benchmarks or created/adapted for the tests in \cite{DBLP:conf/pods/BenediktKMMPST17}.

The goal of the Join Order Benchmark (JOB) \cite{Leis2017} was to evaluate the impact of a good join order on the performance of query evaluation in standard RDBMS. Those queries were formulated over the real-world dataset Internet Movie Database (IMDB). All of the queries have between 3 and 16 joins. Clearly, as the goal was to measure the impact of a good join order, those 33 queries are of higher complexity, hence 7 out of the 33 queries have $\hw \geq 2$.

The 33 TPC-H queries in our benchmark are taken from the GitHub repository 
originally provided by Michael Benedikt and Efthymia Tsamoura \cite{BenediktCQs} for the work on [12].
Out of the 33 CQs based on the TPC-H benchmark \cite{tpch}, 13 queries were handcrafted and 20 randomly generated. The TPC-H benchmark has been widely used to assess multiple aspects of the capabilities of RDBMS to process queries. They reflect common workloads in decision support systems and were chosen to have broad industry-wide relevance.

From SQLShare \cite{shrjainSQLShare}, a multi-year SQL-as-a-service experiment with 
a large set of real-world queries, we extracted 15,170 queries by considering 
all CQs (in particular, no nested SELECTs).
After elimi\-nating trivial queries (i.e., queries with $\leq 2$ atoms, whose
acyclicity is immediate) and duplicates, we ended up with 290 queries.

The random queries were generated with a tool that stems from the work on query answering using views in \cite{2001Pottinger}. 
The query generator allows 3 options: chain/star/random queries. Since the former two types are trivially acyclic, 
we only used the third option. Here it is possible to supply several parameters for the size of the generated queries. In terms of the 
resulting hypergraphs, one can thus fix the number of vertices, number of edges and arity. We have generated 500 CQs with 5 -- 100 vertices, 
3 -- 50 edges and arities from 3 to 20.
These values correspond to the values observed for the {\it CQ Application\/} hypergraphs.
However, even though these size values have been chosen similarly, 
the structural properties of the hypergraphs in two groups {\it CQ Application\/} and {\it CQ Random\/} 
differ significantly,
as will become  clear from our analysis in Section \ref{sect:EmpiricalAnalysis}.

\smallskip
$\bullet$ \
{\it CSPs}: In total, our benchmark currently contains 2,035 hypergraphs from CSP instances, out of which 
1,953 instances were obtained 
from \url{xcsp.org} (see also \cite{xcsp}).
We have selected 
all CSP instances from \url{xcsp.org} 
with less than 100 constraints 
such that all constraints are extensional. 
These instances are divided into CSPs
from concrete applications, called \textit{CSP Application} in the sequel (1,090 instances), 
and randomly generated CSPs, called \textit{CSP Random} below (863 instances). 
In addition, we have included 82 CSP instances from previous hypertree width analyses
provided at \url{https://www.dbai.tuwien.ac.at/proj/hypertree/}; all of these stem from industrial applications and/or further CSP benchmarks. We refer to these instances as {\em other CSPs}.
\nopVLDB{*******************
Out of the 82 instances 61 were already used as a benchmark for the 
HD-computation in \cite{DBLP:journals/jea/GottlobS08}. These 61 instances are split into
15 CSP instances from DaimlerChrysler, 12 CSP instances from Grid2D, and 24 CSP instances from the ISCAS'89 benchmark on circuits. The 82 instances also include 35 instances used in the MaxSAT Evaluation 2017 \cite{berg2017maxsat}. There a MaxSAT encoding was created to determine the $\ghw$ of each of the hypergraphs. During the competition 7 out of the 35 instances were solved. 
We refer to the resulting 82 hypergraphs as ``other CSPs''. In our extended analysis, we keep them separated from the 
instances in the ``CSP Application'' and ``CSP Random'' category mentioned above
because of the different provenance and for the sake of comparability with previous evaluations of hypergraph decomposition algorithms. 
*******************}

\smallskip

Our HyperBench benchmark consists of these instances converted 
to hypergraphs.
In Figure~\ref{fig:hg-sizes}, we show the number of vertices, the number of edges and the arity (i.e., the maximum size of the edges) as 
three important metrics of the size of each hypergraph. 
The smallest are those coming from {\it CQ Application\/} (at most 10 edges), while the 
hypergraphs coming from CSPs can be significantly larger (up to 2993 edges). %
Although some hypergraphs are very big, more than 50\% of all hypergraphs have maximum arity less than 5. 
In Figure~\ref{fig:hg-sizes} we can easily compare the different types of hypergraphs, e.g.\ hypergraphs of arity greater than 20  only exist in the {\em CSP Application} class;
the {\em other CSPs\/} class contains the highest portion of 
hypergraphs with a big number of vertices and edges, etc.
The hypergraphs and the results of their analysis can be accessed through our web tool, available at \url{http://hyperbench.dbai.tuwien.ac.at}.

\nop{******************************************
\begin{itemize}
\item concrete values of the hypergraph properties 
which are crucial for 
the tractable computation of ghw and/or approximation of fhw
(see Section \ref{sect:EmpiricalAnalysis});
\item exact values or at least lower and/or upper bounds on the hw 
of the hypergraphs in our benchmark 
(see also Section \ref{sect:EmpiricalAnalysis})
\item potential improvements of the width via GHDs or 
via fractionally improved HDs (see Sections 
\ref{sect:ghw-implementation} and \ref{sect:fractionally-improved}).
\end{itemize} 
******************************************}

\section{First Empirical Analysis}
\label{sect:EmpiricalAnalysis}

In this section, we present first empirical results obtained with the 
HyperBench benchmark. On the one hand, 
we want to get an overview of the hypertree width of 
the various types of hypergraphs in our benchmark 
(cf.\ Goal 2 in Section~\ref{sect:introduction}).
On the other hand, 
we want to find out how realistic 
the restriction to low values for certain hypergraph invariants is
(cf.\ Goal 3 stated in Section~\ref{sect:introduction}).

\medskip

\noindent
{\bf Hypergraph Properties.}
In \cite{FGP2017tractablefhw,pods/FischlGP18}, 
several invariants of 
hypergraphs were used 
to make the \rec{GHD,\,$k$}  and \rec{FHD,\,$k$} 
problems  tractable or, at least, 
easier to approximate. We thus investigate the following properties (cf.\ Definitions \ref{def:bip}~--~\ref{def:vc}):
\begin{itemize}
 \item {\it Deg}: the degree of the underlying hypergraph
 \item {\it BIP}: the intersection width
 \item $c$-{\it BMIP}: the $c$-multi-intersection width for $c \in \{3,4\}$
 \item {\it VC-dim}: the VC-dimension 
\end{itemize}

\begin{table}[t]
  \centering 
         \caption{Hypergraph properties of all benchmark instances}
  \label{tab:hg-props}%
    \begin{tabular}{rrrrrr}
       \multicolumn{6}{c}{\it CQ Application} \\
    \toprule
        \multicolumn{1}{c}{$i$}  & \multicolumn{1}{c}{Deg} & 
\multicolumn{1}{c}{BIP} & \multicolumn{1}{c}{3-BMIP} & 
\multicolumn{1}{c}{4-BMIP} & \multicolumn{1}{c}{VC-dim} \\
     \midrule
    0     & 0     & 0     & 118   & 173   & 10 \\
    1     & 2     & 421   & 348   & 302   & 393 \\
    2     & 176   & 85    & 59    & 50    & 132 \\
    3     & 137   & 7     & 5     & 5     & 0 \\
    4     & 87    & 5     & 5     & 5     & 0 \\
    5     & 35    & 17    & 0     & 0     & 0 \\
    6     & 98    & 0     & 0     & 0     & 0 \\ \bottomrule
    \end{tabular}%
    \;
    \begin{tabular}{rrrrrr}
       \multicolumn{6}{c}{\it CQ Random} \\
    \toprule
        \multicolumn{1}{c}{$i$}  & \multicolumn{1}{c}{Deg} & 
\multicolumn{1}{c}{BIP} & \multicolumn{1}{c}{3-BMIP} & 
\multicolumn{1}{c}{4-BMIP} & \multicolumn{1}{c}{VC-dim} \\
     \midrule
    0     & 0     & 1     & 16    & 49    & 0 \\
    1     & 1     & 17    & 77    & 125   & 20 \\
    2     & 15    & 53    & 90    & 120   & 133 \\
    3     & 38    & 62    & 103   & 74    & 240 \\
    4     & 31    & 63    & 62    & 42    & 106 \\
    5     & 33    & 71    & 47    & 28    & 1 \\
    6     & 382   & 233   & 105   & 62    & 0 \\ \bottomrule
    \end{tabular}%
    \;
    \begin{tabular}{rrrrrr}
    \multicolumn{6}{c}{\it CSP Application \& Other} \\
    \toprule
        \multicolumn{1}{c}{$i$}  & \multicolumn{1}{c}{Deg} & 
\multicolumn{1}{c}{BIP} & \multicolumn{1}{c}{3-BMIP} & 
\multicolumn{1}{c}{4-BMIP} & \multicolumn{1}{c}{VC-dim} \\
     \midrule
    0     & 0     & 0     & 597   & 603   & 0 \\
    1     & 0     & 1037  & 495   & 525   & 0 \\
    2     & 597   & 95    & 57    & 23     & 1115 \\
    3     & 6     & 29     & 21     & 21     & 52 \\
    4     & 20     & 10     & 2     & 0     & 0 \\
    5     & 6     & 0     & 0     & 0     & 0 \\
    $>$5   & 543   & 1     & 0     & 0     & 0 \\
    \bottomrule
    \end{tabular}
    \;
    \begin{tabular}{rrrrrr}
    \multicolumn{6}{c}{\it CSP Random} \\
    \toprule
          \multicolumn{1}{c}{$i$}  & \multicolumn{1}{c}{Deg} & 
\multicolumn{1}{c}{BIP} & \multicolumn{1}{c}{3-BMIP} & 
\multicolumn{1}{c}{4-BMIP} & \multicolumn{1}{c}{VC-dim} \\
    \midrule
    0     & 0     & 0     & 0     & 0     & 0 \\
    1     & 0     & 200   & 200   & 238   & 0 \\
    2     & 0     & 224   & 312   & 407   & 220 \\
    3     & 0     & 76    & 147   & 95    & 515 \\
    4     & 12    & 181   & 161   & 97    & 57 \\
    5     & 8     & 99    & 14    & 1     & 71 \\
    $>$5    & 843   & 83    & 29    & 25    & 0 \\
    \bottomrule
    \end{tabular}%
\end{table}%

The results obtained from computing 
{\it Deg}, {\it BIP}, $3$-{\it BMIP},  $4$-{\it BMIP}, and {\it VC-dim}
for the hypergraphs in the HyperBench benchmark 
are shown in Table~\ref{tab:hg-props}.

Table~\ref{tab:hg-props} has to be read as follows: In the first column, we distinguish different values of the various hypergraph metrics. 
In the columns labelled ``Deg``, ``BIP``, etc., we indicate for how many instances each metric has a particular value. 
For instance, by the last row in the second column, 
only 98 non-random CQs
have degree~$>5$. Actually, for most CQs, the degree is less than 10.
Moreover, for the BMIP, already with intersections of 3 edges, we get $\cmiwidth{3}{\HH} \leq 2$ 
for almost all non-random CQs. Also the VC-dimension is at most~2.

For CSPs, all properties may have higher values. However, we note a significant 
difference between randomly generated CSPs and the rest: 
For hypergraphs in the groups {\em CSP Application\/} and {\em CSP Other\/}, 
543 (46\%) hypergraphs have a high degree ($>$5), but nearly all instances have BIP or BMIP of less than 3. And most instances have a VC-dimension of at most 2. 
In contrast, nearly all random instances have a 
significantly higher degree (843 out of 863 instances with a degree $>$5). 
Nevertheless, many instances have small BIP and 
BMIP. For nearly all hypergraphs (838 out of 863) 
we have $\cmiwidth{4}{\HH} \leq 4$. 
For 5 instances the computation of the VC-dimension
timed out. For all others, the VC-dimension is $\leq 5$ for random CSPs.
Clearly, as seen in Table \ref{tab:hg-props}, the random CQs resemble the random CSPs a lot more than the CQ and CSP Application instances. For example, random CQs have similar to random CSPs high degree (382 (76\%) with degree $>5$), higher BIP and BMIP. Nevertheless, similar to random CSPs, the values for BIP and BMIP are still small for many random CQ instances.

To conclude, 
for 
the proposed properties, in particular BIP/BMIP and VC-dimension,
most hypergraphs 
in 
our
benchmark
(even for non-random CQs and CSPs) 
indeed  
have low values. 

\medskip
\begin{figure}
\centering 
     \fbox{\includegraphics[width=0.45\textwidth]{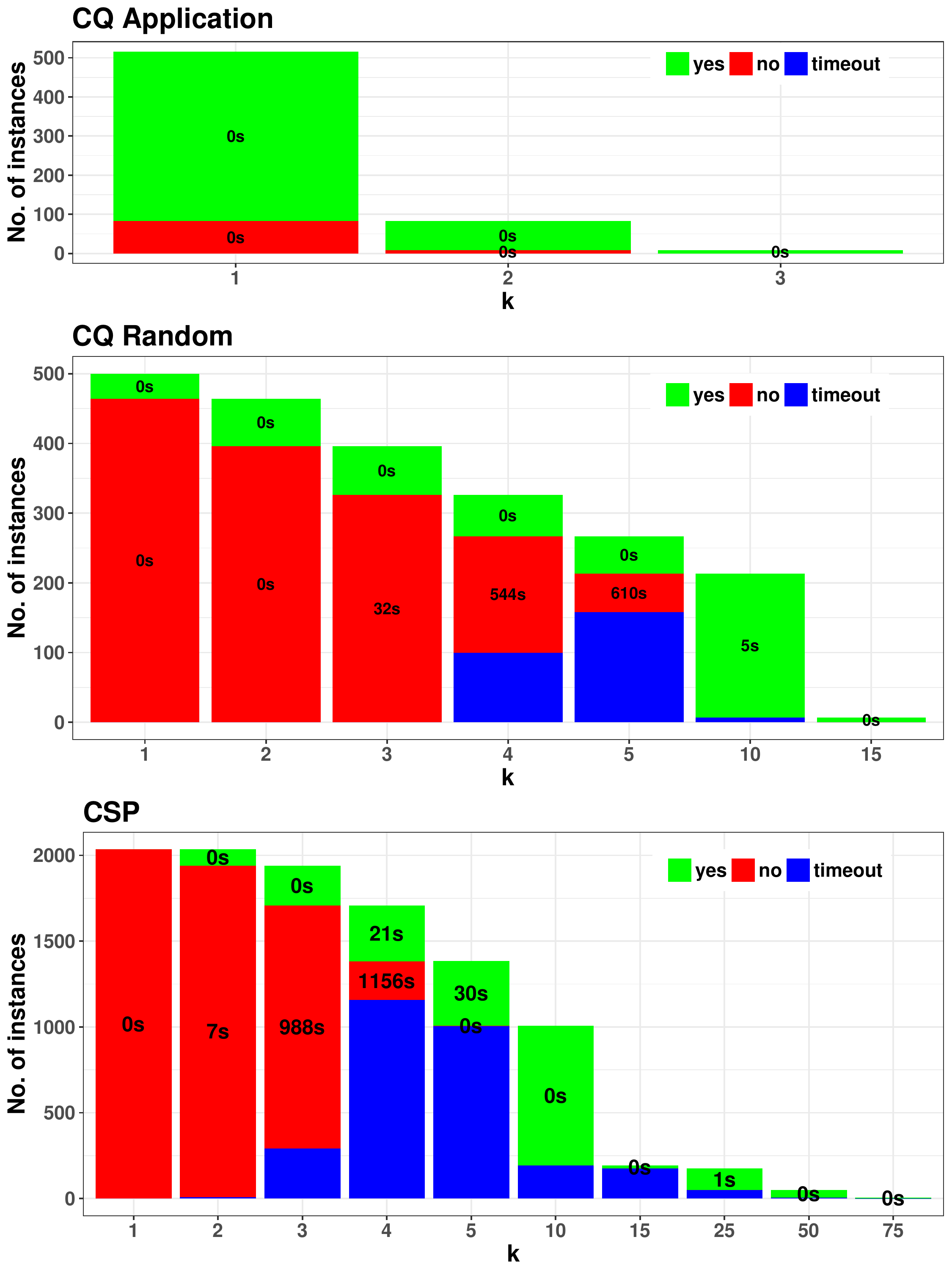}}
    \caption{HW analysis (labels are average runtimes in s)} 
\label{fig:hw}      
\end{figure}

\noindent
{\bf Hypertree Width.} We have systematically applied the 
$\hw$-compu\-tation from \cite{DBLP:journals/jea/GottlobS08} 
to all hypergraphs 
in the benchmark. The results are summarized in Figure~\ref{fig:hw}. 
In our experiments, we proceeded as follows. 
We distinguish between 
{\it CQ Application\/}, {\it CQ Random\/},  
and all three groups of CSPs taken together.
For every hypergraph $H$, we first tried to solve the 
\rec{HD,\,$k$} problem for $k = 1$. In case of {\em CQ Application\/}, we thus got 454 yes-answers 
and 81 no-answers. The number in each bar indicates the average runtime to 
find these yes- and no-instances, respectively. 
Here, the average runtime was ``0'' (i.e., less than 1 second) 
in both cases. 
For {\it CQ Random\/} we got 36 yes- and 464 no-instances with an average runtime below 1 second.
For all %
CSP-instances, we only got no-answers. %

In the second round, we tried to solve the \rec{HD,\,$k$} problem for 
$k = 2$ for all hypergraphs that yielded a no-answer for $k=1$. 
Now the picture is a bit more diverse: 73 of the remaining 81 CQs from 
{\it CQ Application\/} yielded a 
yes-answer in less than 1 second. For the hypergraphs stemming from {\it CQ Random} (resp.\ CSPs), only 
68 (resp.\ 95) instances yielded a yes-answer (in less than 1 second on average), while 396 
(resp.\ 1932) instances yielded a no-answer in less than 7 seconds on average and 8 CSP instances 
led to a timeout (i.e., the program did not terminate within 3,600 seconds).

This procedure is iterated by incrementing $k$ and running the $\hw$-computation for all instances, that either yielded a no-answer or a timeout in the previous round. For instance, for queries from {\it CQ Application\/}, 
one further round is needed after the second round. 
In other words, we confirm the observation of low $\hw$, which was 
already made for CQs of arity $\leq 3$ in \cite{DBLP:journals/pvldb/BonifatiMT17,DBLP:conf/sigmod/PicalausaV11}.
For the hypergraphs stemming from {\em CQ Random} (resp.\ CSPs), 396 (resp.\ 1940 )instances 
are left in the third round, of which 70 (resp.\ 232) yield a yes-answer in less than 1 second on average, 326 (resp.\ 1415) instances yield a no-answer in 32 (resp.\ 988) seconds on average and 
no (resp.\ 293) instances yield a timeout. 
Note that, as we increase $k$, 
the average runtime  and the percentage of timeouts first increase up to a certain point and then they decrease.
This is due to the fact that,
as we increase $k$, the number of combinations of edges to be considered in 
each $\lambda$-label (i.e., the function $\lambda_u$ at each node $u$ 
of the decomposition) increases. In principle, we have to test 
$\mathcal{O} (n^{k})$ combinations, where $n$ is the number of edges.
However, if $k$ increases beyond a certain point, then it gets easier to ``guess'' a $\lambda$-label since an increasing portion of the 
$\mathcal{O} (n^{k})$ possible combinations leads to a solution (i.e., 
an HD of desired width).

To answer the question in {\it Goal 2}, it is indeed the case that for a big number of instances, 
the hypertree width is small enough to allow for efficient evaluation of CQs or CSPs: all instances of non-random CQs have $\hw \leq 3$ no matter 
whether their arity is bounded by 3 (as in case of SPARQL queries) or not; and a large portion (at least 1027, i.e., ca.\ 50\%) 
of all 2035 CSP instances have $\hw \leq 5$. 
In total, including random CQs, 1,849 (60\%) out of 3,070 instances have 
$\hw \leq 5$, for which
we could determine the exact hypertree width for 1,453 instances; 
the others may even have lower $\hw$.

\medskip

\noindent
{\bf Correlation Analysis.} Finally, we have analysed the pairwise correlation between 
all properties. 
Of course, the different intersection widths (BIP, 3-BMIP, 4-BMIP)
are highly correlated.
Other than that, we only observe quite a high correlation of the arity with the number of vertices and the hypertree width and 
of the number of vertices with the arity and the hypertree width. 
Clearly, the correlation between arity and hypertree width 
is mainly due to the CSP instances and the random CQs since, for non-random CQs, 
the $\hw$ never increases beyond~$3$,  independently of the arity.

A graphical presentation of all pairwise correlations is given in Figure~\ref{fig:corr}. Here, large, dark circles indicate a high correlation, 
while small, light circles stand for low correlation. 
Blue circles indicate a positive  correlation while red circles stand for a negative correlation.  
In  \cite{pods/FischlGP18}, we have argued that Deg, BIP, 3-BMIP, 4-BMIP and VC-dim are non-trivial restrictions to achieve tractability. It is interesting to note that, 
according to the correlations shown in Figure~\ref{fig:corr}, 
these properties have almost no impact on the hypertree width of our hypergraphs.
This underlines the usefulness of these restrictions in the sense that (a) they make the GHD computation and 
FHD approximation easier \cite{pods/FischlGP18} but (b) low values of degree, (multi-)intersection-width, or VC-dimension do not pre-determine low values of the widths.

\begin{figure}[htbp]
\centering 
     \fbox{\includegraphics[width=0.45\textwidth]{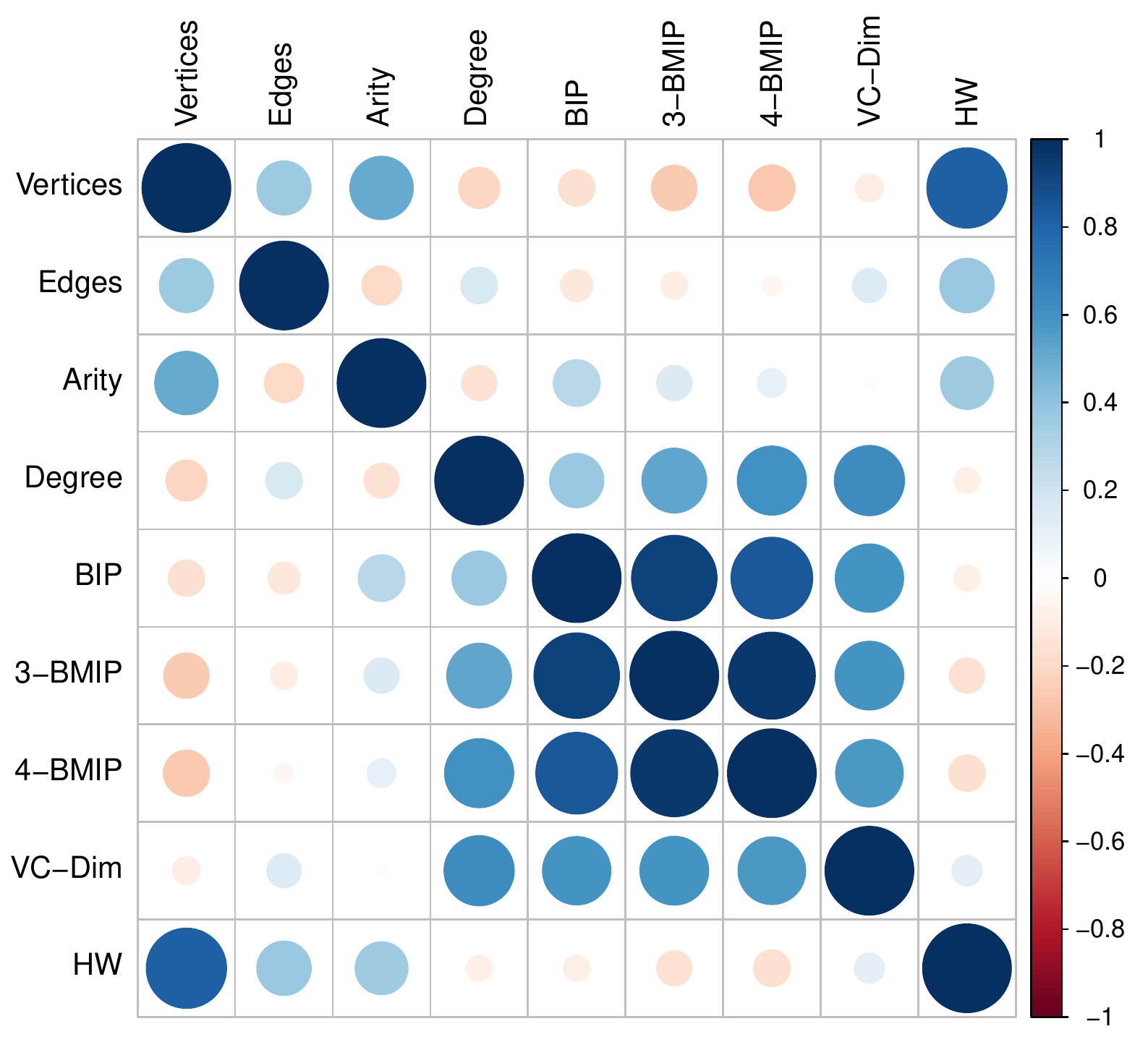}}    
    \caption{Correlation analysis.} 
\label{fig:corr}     
\end{figure}  

\section{GHW Computation}
\label{sect:ghw-implementation}

In this section, we report on new algorithms and implementations to solve the \rec{GHD,\,$k$} problem and on new empirical results.

\medskip

\noindent
{\bf Background.}
In \cite{pods/FischlGP18}, it is shown that the \rec{GHD,\,$k$} problem becomes tractable for fixed $k \geq 1$, if
we restrict ourselves to a class of hypergraphs enjoying the BIP. As our first empirical analysis with the HyperBench has shown 
(see Section \ref{sect:EmpiricalAnalysis}), 
it is indeed realistic to assume that the 
intersection width of a given hypergraph is small. We have therefore extended the 
$\hw$-computation from \cite{DBLP:journals/jea/GottlobS08}  by an implementation of the \rec{GHD,\,$k$} algorithm 
from \cite{pods/FischlGP18},
which will be referred to as the ``$\ghw$-algorithm'' in the sequel. 
This algorithm
is parameterized, so to speak, by two integers:  $k$ (the desired width of a GHD) and  
$i$ (the intersection width of $H$). 

The key idea of the $\ghw$-algorithm is to add a polynomial-time computable set $f(H,k)$ of 
subedges of edges in $E(H)$ to the hypergraph $H$, such that $\ghw(\HH)=k$ 
iff 
$\hw(\HH')=k$ with 
$H = (V(H), E(H))$ and
$H' = (V(H), E(H) \cup f(\HH,k))$. 
Tractability of \linebreak \rec{GHD,\,$k$} follows immediately from the 
tractability of the \rec{HD,\,$k$} problem. 
The set $f(H,k)$ is defined as 
$$f(\HH,k)= \bigcup_{e\in E(\HH)}   
\Big({\bigcup_{e_1,\ldots,e_j\in (E(\HH) \setminus \{e\}),\, j\leq k }} 
2^{(e\cap (e_1\cup\cdots\cup e_j))}\Big),$$
i.e., $f(\HH,k)$ contains all subsets of intersections of edges $e \in E(H)$ with unions of $\leq k$ 
edges of $H$ 
different from $e$. By the BIP, the intersection $e\cap (e_1\cup\cdots\cup e_j)$ has at most $i\cdot k$ 
elements. Hence, for fixed constants $i$ and $k$,  $|f(H,k)|$ is polynomially bounded.

\medskip

\noindent
{\bf ``Global'' implementation.}
In a straightforward implementation of this algorithm, we compute $f(\HH,k)$ and from this 
$H'$ and call the $\hw$-compu\-tation from \cite{DBLP:journals/jea/GottlobS08}  for the 
\rec{HD,\,$k$} problem as a ``black box''. 
A coarse-grained overview of the results is given in Table~\ref{tab:ghw_instances_solved} in the column labelled as `GlobalBIP''.
We call this implementation of the $\ghw$ algorithm of \cite{pods/FischlGP18} ``global'' to indicate that the set $f(H,k)$ is computed ``globally'', once and for all, for the {\em entire\/} hypergraph.
We have run the program on each hypergraph from the HyperBench 
up to hypertree width $6$, 
trying to get a smaller $\ghw$ than $\hw$. 
We have thus run the $\ghw$-algorithm with the following parameters: for all hypergraphs $H$ 
with $\hw(H) = k$ (or $\hw \leq k$ and, due to timeouts, we do not know if $\hw \leq k-1$ holds), 
where $k \in \{3,4,5,6\}$, 
try to solve the \rec{GHD,\,$k-1$} problem. 
In other words, we just tried to improve the width by 1. 
Clearly, for $\hw(H) \in \{1,2\}$, no improvement is possible since, in this case,  
$\hw(H) = \ghw(H)$ holds.

In Table~\ref{tab:ghw_instances_solved}, we report on the number of 
``successful'' attempts to solve the \rec{GHD,\,$k-1$} 
problem for hypergraphs with $\hw =k$. Here ``successful'' means that 
the program terminated within 1 hour. 
For instance, for the 310 hypergraphs with $\hw = 3$ 
in the HyperBench, 
the ``global'' computation terminated in 128 cases (i.e., 41\%) when trying to solve \rec{GHD,\,$2$}. The average runtime of these ``successful'' runs was 537 seconds.
For the 386 hypergraphs with $\hw = 4$, 
the ``global'' computation terminated in 137 cases (i.e., 35\%) 
with average runtime 2809 
when trying to solve the \rec{GHD,\,$3$} problem.
For the 886 hypergraphs with $\hw \in \{5,6\}$, the ``global'' computation only 
terminated in 13 cases (i.e., 1.4\%). 
Overall, it turns out that the set $f(\HH,k)$ may be very big (even though it is 
polynomial if $k$ and $i$ are constants). Hence, $H'$ can become  considerably bigger than $H$. 
This explains the frequent timeouts in the GlobalBIP column in Table~\ref{tab:ghw_instances_solved}.

\begin{table}
  \centering 
    \caption{%
    Comparison of GHW algorithms w. avg. runtime (s)}
  \label{tab:ghw_instances_solved}%
        \setlength{\tabcolsep}{2.4pt}
     \begin{tabular}{lrrrrrrr} 
       \toprule
        \multicolumn{1}{c}{\small $\hw \ra$} &  & \multicolumn{2}{c}{\globalbip} & 
\multicolumn{2}{c}{\localbip} & \multicolumn{2}{c}{\balsep} \\

\multicolumn{1}{c}{\small $\ghw$} & \multicolumn{1}{c}{\small\it total} & \multicolumn{1}{c}{\small\it yes} & 
\multicolumn{1}{c}{\small\it no} & \multicolumn{1}{c}{\small\it yes} & 
\multicolumn{1}{c}{\small\it no} &\multicolumn{1}{c}{\small\it yes} & 
\multicolumn{1}{c}{\small\it no}  \\
     \midrule
     $3 \ra 2$ & 310 & \multicolumn{1}{c}{-} & 128 (537)  & \multicolumn{1}{c}{-} & 195 (162)  & \multicolumn{1}{c}{-} & 307 (12) \\
     $4 \ra 3$ & 386  & \multicolumn{1}{c}{-} & 137 (2809)  & \multicolumn{1}{c}{-} & 54  (2606) & \multicolumn{1}{c}{-}  & 249  (54) \\
     $5 \ra 4$ &  427 & \multicolumn{1}{c}{-} & \multicolumn{1}{c}{-}  & \multicolumn{1}{c}{-} & \multicolumn{1}{c}{-} & \multicolumn{1}{c}{-} & 148 (13)  \\ 
    $6 \ra 5$ & 459 & 13 (162)  & \multicolumn{1}{c}{-} & 13  (60)  & \multicolumn{1}{c}{-} & \multicolumn{1}{c}{-} & 180 (288)  \\ 
     \bottomrule
    \end{tabular}%
\end{table}%

\medskip

\noindent
{\bf ``Local'' implementation.}
Looking for ways to improve the $\ghw$-algorithm, we closely inspect the role played by the set $f(\HH,k)$
in the tractability proof in \cite{pods/FischlGP18}.
The definition of this set is motivated by
the problem that, in the top down construction of a GHD, we may want to choose at some node $u$
the bag $B_u$ such that $x \not\in B_u$ for some variable $x \in B(\lambda_u) \cap V(T_u)$. 
This violates condition (4) of Definition~\ref{def:HD} (the ``special condition'') and is therefore forbidden in an HD. 
In particular, there exists an edge $e$ with $x \in e$ and $\lambda_u(e) = 1$. 
The crux of the $\ghw$-algorithm in \cite{pods/FischlGP18} is that for every such ``missing'' variable $x$, 
the set $f(\HH,k)$ contains a subedge $e' \subseteq e$ with $x \not \in e'$. Hence, replacing $e$ by $e'$ in $\lambda_u$ 
(i.e., setting $\lambda_u(e) = 0$, $\lambda_u(e') = 1$ and leaving $\lambda_u$ unchanged elsewhere) eliminates the special condition 
violation. By the connectedness condition, it suffices to consider the intersections of $e$
with unions of edges that may possibly occur in bags of $T_u$ rather than with arbitrary edges in $E(H)$.
In other words, for each node $u$ in the decomposition, we may restrict $f(H,k)$ to an appropriate subset 
$f_u(H,k) \subseteq f(H,k)$.

The results obtained with this enhanced version of the $\ghw$-computation are shown in 
Table~\ref{tab:ghw_instances_solved} in the column labelled ``LocalBIP''. 
We call this implementation of $\ghw$-computation ``local'' because the set 
$f_u(H,k)$ of subedges of $H$ to be added to the hypergraph
is computed separately for each node $u$ of the decomposition.
Recall that in this table, the ``successful'' calls of the program are recorded.
Interestingly, for the 
hypergraphs with $\hw = 3$, the  ``local'' computation performs significantly better (namely 63\% solved
with average runtime 162 seconds rather than 41\% with average runtime 537 seconds). 
In contrast, for the hypergraphs with $\hw = 4$, the ``global'' computation is significantly more successful.
For $\hw \in \{5,6\}$, the ``global'' and ``local'' 
computations are
equally bad.
A possible explanation for the reverse behaviour of ``global'' and ``local'' computation in case of 
$\hw = 3$ as opposed to $\hw = 4$ is that the restriction of the ``global'' set 
$f(H,k)$ of subedges to the ``local'' set $f_u(H,k)$ at each node $u$ seems to be
quite effective for the hypergraphs with $\hw = 3$. In contrast, 
the additional cost of having to compute $f_u(H,k)$ at each node $u$ becomes counter-productive, when the set of subedges thus eliminated is 
not significant. It is interesting to note that the sets of solved instances 
of the global computation and the local computation are
incomparable, i.e., in some cases one method is better, while
in other cases the other method is better.

\begin{figure}[t]
\fbox{
\begin{minipage}{\textwidth}
\begin{small}
\begin{center} 
\begin{tabbing}
xx \= xxx \= xxx \= xxx \= xxx \= xxx \= xxx \= xxx \= \kill
{\bf ALGORITHM} Find\_GHD\_via\_balancedSeparators \\
// high-level description  \\[1.1ex]
{\bf Input:} \hskip 7pt  hypergraph $H'$, integer $k\geq 0$. \\
{\bf Output:}  a GHD $\left<T,B_u,\lambda_u\right>$ of width $\leq k$ if exists, \\ 
\hskip 32pt   ``Reject'', otherwise.  \\[1.2ex]
 
{\bf Procedure} Find\_GHD ($H$: Hypergraph, $\Sp$: Set of special edges) \\
{\bf begin}   \\
1.\ \> Base Case: if there are only special edges left and $|\Sp| \leq 2$ then \\
\> \> stop and return a GHD with one node for each special edge. \\[1.1ex]
2.\ \> Find a balanced separator: for all functions $\lambda \colon E(H') \ra \{0,1\}$ \\
\>\> check if $\lambda_u$ is a balanced separator for $H$; \\
\> \> if none is found  then return Reject. \\[1.1ex]
3.\ \> Split $H$ into connected components $C_1, \dots, C_\ell$ w.r.t.\ $\lambda_u$: \\
  \>\> $C_i \subseteq V(H) \setminus B(\lambda_u)$ for every $i$ and $C_i$  is connected in $H$\\
  \>\> and each $C_i$ is maximal with this property.
\\[1.1ex]
4.\ \> Build the pair $\left<H_i,\Sp_i\right>$ (the subhypergraph based on $C_i$ and \\
\> \> the special edges in $C_i$) for each connected component $C_i$; \\
\> \> add $B(\lambda_u)$ as one more special edge to each set $\Sp_i$.
\\[1.1ex]
5.\ \> Call Find\_GHD($H_i$, $\Sp_i$) for each pair  $\left<H_i,\Sp_i\right>$; \\
\> \> each successful call returns a GHD $T_i$ for $H_i$ \\
\> \> if one call returns Reject then return Reject. \\[1.1ex]
6.\ \> Create and return a new GHD for $H$ having $\lambda_u$ as root: \\
\>\> each $T_i$ has one leaf node labelled $B(\lambda_u)$; \\
\>\> the new GHD is obtained by gluing together all subtrees $T_i$ \\
\>\> at the node with label $B(\lambda_u)$. \\
{\bf end}  \\[1.1ex]
{\bf begin} (* Main *)  \\
  \> {\bf  return} Find\_GHD ($H'$, $\emptyset$); \\
{\bf end}  

\end{tabbing}
\end{center}
\end{small}
\end{minipage}
}

\vspace{-1.8ex}
\caption{Recursive GHD-algorithm via balanced separators}
 \label{fig:rec-algo}
 \end{figure}  

\medskip

\noindent
{\bf New alternative approach: ``balanced separators''.}
We now propose a completely new approach, based on so-called ``balanced separators''. The latter are a familiar concept in graph theory \cite{DBLP:conf/stoc/FeigeM06,DBLP:conf/wea/SchildS15} -- denoting a set $S$ of vertices  of a graph $G$, such that the subgraph $G'$ 
induced by $V(G) \setminus S$ has no connected component larger than some given size, e.g., $\alpha \cdot |V|$ for some given $\alpha \in (0,1)$. In our setting, we may consider the label $\lambda_u$ at some 
node $u$ in a GHD as {\em separator\/} in the sense that we can consider connected components of the subhypergraph $H'$ of $H$ induced by $V(H) \setminus B_u$. Clearly, in a GHD, we may consider any node as the root. So suppose that $u$ is the root of some GHD. Moreover, as is shown in \cite{pods/FischlGP18} 
in  the proof of tractability of \rec{GHD,\,$k$} in case of the BIP, we may choose $\lambda_u$ such that 
$B(\lambda_u) = B_u$ if the subedges in $f(\HH,k)$  have been added to the hypergraph. 
\nopVLDB{*************
This is due to the fact that $B_u \subseteq B(\lambda_u)$ holds by the definition of GHDs and, conversely, if $\lambda_u$ contains
an edge $e$ with $e \not\subseteq B_u$, 
then we may replace $e$ in $\lambda_u$ by the subedge
$e' = e \cap B_u$, which is guaranteed to be contained in $f(\HH,k)$. 
*************}

By the HD-algorithm from \cite{2002gottlob}, we know that an HD of $H'$ (and, hence, a GHD of $H$) can 
be constructed in such a way that every subtree rooted at a child node $u_i$ of $u$ contains only one connected 
component $C_i$ of the subhypergraph of $H'$ induced by $V(H) \setminus B_u$. 
For our purposes, it is convenient to define the size of a component $C_i$ as the number of edges that have to be covered at some node in the subtree rooted at $u_i$ in the GHD. %
We thus call a separator $\lambda_u$ ``balanced'', if the size of each component $C_i$ is at most $|E(H')| / 2$.
The following observation is immediate: 

\begin{proposition}
In every GHD, there exists a node $u$ (which we may choose as the root) such 
that $\lambda_u$ is a balanced separator.
\end{proposition}

This property allows us to design the algorithm sketched in 
Figure~\ref{fig:rec-algo} to compute a GHD of $H'$. 
Actually, as will become clear below, we assume that the input to this recursive algorithm consists of a hypergraph plus a set 
$\Sp$ of ``special edges'' and we request that the GHD to be constructed contains ``special nodes'', which (a) have to be leaf nodes in the decomposition and (b) the $\lambda$-label of such a leaf node consists of a single special edge only. 
Each special edge contains the set of vertices $B_u$ of some balanced separator $\lambda_u$ 
further up in the hierarchy of recursive calls of the decomposition algorithm. 
The special edges are propagated to the recursive calls for subhypergraphs in order to 
determine how to assemble the
overall GHD from 
the GHDs of the subhypergraphs. This will become clearer in the proof sketch of 
Theorem~\ref{theo:correct-recursive}.

\begin{theorem}
\label{theo:correct-recursive}
Let $H$ be a hypergraph, let $k \geq 1$, and let $H'$ be obtained from $H$ by adding the subedges in
$f(H,k)$ to $E(H)$. Then the algorithm Find\_GHD\_via\_balancedSeparators given in 
Figure~\ref{fig:rec-algo} outputs a GHD of width $\leq k$ if one exists and rejects otherwise.
\end{theorem}

\begin{proof}[Proof Sketch]
Steps 1 --5 of the algorithm in Figure~\ref{fig:rec-algo} 
essentially correspond to the computation of $\lambda_u$ and $B_u$ for the root node $u$ in the HD-computation
of \cite{2002gottlob}. The most significant modifications here are due to the handling of ``special edges'' in parameter $Sp$. 
A crucial property of the construction in \cite{2002gottlob} and also of our construction here is that each subtree below node $u$ in the decomposition only contains vertices from a single 
connected component $C_i$ w.r.t. $V(H) \setminus B_u$ (see Steps 3 and 4). Since the special edges come from 
such bags $B_u$, special edges can never be used as separators in recursive calls below. Hence, the base case (in Step~1) is reached for 
$E(H) = \emptyset$ and $|Sp| \leq 2$. Indeed, $|Sp| \geq 3$ cannot occur because then one of the special edges would have to be
a separator of the remaining special edges. Moreover, we can exclude special edges from the search for a balanced separator (in Step 2).

The correctness of assembling a GHD (in Step 6) from the results of the recursive calls
can be shown by structural induction on the
tree structure of a GHD: suppose that the recursive calls in the 
algorithm for each hypergraph $H_i$ with set $\Sp_i$ of special edges are correct, i.e., they yield for each hypergraph $H_i$ a GHD ${\cal D}_i$
such that each special edge $s$ in $\Sp_i$ is indeed covered by a leaf node 
in ${\cal D}_i$ whose $\lambda$-label consists of $s$ only.  In particular, 
since $s = B(\lambda_u)$ is a special edge contained in $\Sp_i$ for each $i$, 
there exists a leaf node $t_i$ in  ${\cal D}_i$ with $\lambda_{t_i} = \{s\}$. 
In a GHD, any node can be taken as the root. We thus choose
$t_i$ as the root node in each GHD ${\cal D}_i$. By construction, we have 
$B_{u} = B_{t_1} = \dots = B_{t_\ell}$. Moreover, any two subhypergraphs $H_i$, $H_j$ contain the vertices from two different connected components. Hence, apart from the vertices  contained in the special edge $s$, any two GHDs
${\cal D}_i$, ${\cal D}_j$  with $i \neq j$ 
have no vertices in common.  We can therefore construct a 
GHD ${\cal D}$ of $H'$ by deleting the root node $t_i$ from each 
GHD ${\cal D}_i$ and by appending the child nodes of each $t_i$ directly as 
child nodes of $u$. Clearly, the connectedness condition is satisfied in the 
resulting decomposition. 
\end{proof}

\nopVLDB{**********************
For the {\em runtime\/} of this recursive algorithm, we make two observations: first, the recursion depth is logarithmically 
bounded by the number of edges in $H'$. And second, the number of {\em balanced\/} separators to be inspected at each step is typically much smaller than the number of all possible separators, as 
the following example illustrates. 

\begin{example}
\label{ex:balSep}
Let $H$ be a hypergraph which is a cycle of size $n$ 
with $V(H) = \{v_1, \dots, v_n\}$ and 
$E(H) = \{ \{v_1,v_2\}$, $\dots,$ 
$\{v_{n-1},v_n\}, \{v_{n},v_1\}\}$.
Moreover, suppose that we are interested in sets $S$ of edges with cardinality 
$k = 2$ to split the cycle into two disconnected paths. 
Clearly, every choice $S$ of two edges 
that leave at least one node ``distance'' between them 
separates $H$ into two paths, i.e., $S = \{ \{v_{i},v_{i+1}\}$, $\{v_{j},v_{j+1}\}\}$
with $j > i+1$ and if $j=n-1$ then $i > 1$ and if $j=n$ then $i > 2$. 
So in total, there exist $\Theta(n^{2})$ separators. For $k = 2$, 
this bound is equal to $\Theta (n^{k})$. 
On the other hand, we get a balanced separator only by choosing two edges which are maximally
(or at least ``almost'' maximally) apart
from each other. Hence, there exist only $\mathcal{O} (n)$ balanced separators. For $k = 2$, this is equal to $\mathcal{O} (n^{k/2})$. 
\end{example}
**********************}

If we look at the number of solved instances in Table~\ref{tab:ghw_instances_solved}, 
we see that the 
recursive algorithm via balanced separators (reported in the last column labelled \balsep) has the least number 
of timeouts due to the fast identification of negative instances (i.e., those with no-answer), 
where it often detects quite fast that a given hypergraph 
does not have a balanced separator of desired~width.
As $k$ increases, the performance of the balanced separators approach deteriorates. 
This is due to $k$ in the exponent of the running time of our algorithm, i.e.\ we need to check 
for each of the possible $\mathcal{O} (n^{k+1})$ 
combinations of $\leq k$ edges if it constitutes a balanced separator. 
\begin{table}
  \centering 
    \caption{GHW of instances with average runtime in s}
  \label{tab:ghw}%
        \begin{tabular}{lrrr}
    \toprule
        \multicolumn{1}{c}{$\hw \ra \ghw$}  & \multicolumn{1}{c}{yes} & 
\multicolumn{1}{c}{no} & \multicolumn{1}{c}{timeout} \\
     \midrule
    $3 \ra 2$     & 0   & 309 (10)    & 1   \\
     $4 \ra 3$    & 0   & 262 (57)    & 124   \\
    $5 \ra 4$     & 0   & 148 (13)    & 279   \\
    $6 \ra 5$     & 18 (129)   & 180 (288)    & 261   \\
     \bottomrule
    \end{tabular}%
\end{table}%
 
\medskip

\noindent
{\bf Empirical results.}
We now look at Table \ref{tab:ghw}, where we
report for all hypergraphs with $\hw \leq k$ and $k \in  \{3,4,5,6\}$, whether 
$\ghw \leq k-1$ could be verified. 
To this end, we run our three algorithms (``global'', ``local'', and ``balanced separators'') in parallel and stop the computation, as soon as one terminates
(with answer ``yes'' or ``no''). The number in parentheses refers to the average runtime 
needed by the fastest of the three algorithms in each case. A timeout occurs if none of the three algorithms terminates within 3,600 seconds.
It is interesting to note that in the vast majority of cases, no improvement of the width is possible when we switch from $\hw$ to $\ghw$: in 98\% of the 
solved cases and 57\% of all instances with $\hw \leq 6$, 
$\hw$ and $\ghw$ have identical values. Actually, we think that the high percentage of the {\em solved cases\/} 
gives a more realistic picture than the percentage of {\em all cases\/}
for the following reason: our algorithms (in particular, the ``global'' and ``local'' computations) need particularly long time for negative instances. This is due to the fact that in a negative case, ``all'' possible choices of $\lambda$-labels for a node $u$ in the GHD have to be tested before we can be sure that no GHD of $H$ (or, equivalently, no HD of $H'$) of desired width exists. 
Hence, it seems plausible that the timeouts are mainly due to negative instances.
This also explains why our new GHD algorithm in Figure~\ref{fig:rec-algo}, which is 
particularly well suited for negative instances,
has the least number of timeouts.

We conclude this section with a final observation: in Figure~\ref{fig:hw}, we 
had many cases, for which only some upper bound $k$ on the $\hw$ could be determined, namely those cases, where the attempt to solve \rec{HD,\,$k$}\
yields a yes-answer and the attempt to solve \rec{HD,\,$k-1$} gives a timeout.
In several such cases, we could get (with the balanced separator approach) a 
no-answer for the \rec{GHD,\,$k-1$} problem, which implicitly gives a 
no-answer for the problem \rec{HD,\,$k-1$}. 
In this way, the alternative approach to the $\ghw$-computation is also profitable for the $\hw$-computation:
for 827 instances with $\hw \leq 6$, we were not able to determine the exact hypertree width. 
Using our new $\ghw$-algorithm, we closed this gap for 297 instances; for these instances $\hw = \ghw$ holds. 

To sum up, we now have a total of 1,778 (58\%) instances for which we determined the exact hypertree width and a total of 1,406 instances (46\%)
for which we determined the exact generalized hypertree width. Out of these, 1,390 instances had 
identical values for $\hw$ and $\ghw$.
In 16 cases, we found an improvement of the width by 1 when moving from $\hw$ to $\ghw$, namely from $\hw = 6$ 
to $\ghw = 5$. In 2 further cases, we could show $\hw \leq 6$ and $\ghw \leq 5$, but the attempt to check 
$\hw = 5$ or $\ghw = 4$ led to a timeout. 
Hence, in response to {\it Goal 6}, $\hw$ is equal to $\ghw$ in 45\% of the cases if we consider all instances
and in 60\% of the cases (1,390 of 2,308) with small width ($\hw \leq 6$). However, if we consider the fully 
solved cases (i.e., where we have the precise value of $\hw$ and $\ghw$), then $\hw$ and $\ghw$ coincide in 
99\% of the cases (1,390 of 1,406).

\section{Fractionally Improved Decompositions}
\label{sect:fractionally-improved}

The algorithms proposed in the literature for computing FHDs are very expensive. 
For instance, even the algorithm used for the tractability result in \cite{FGP2017tractablefhw} 
for hypergraphs of low degree is problematical since it involves a double-exponential factor in the degree. 
Therefore, we investigate the potential of a simplified method to compute approximated FHDs.
Below, we present two algorithms for such approximated FHD computations -- with a trade-off between computational cost and quality of the approximation.
	
\smallskip
$\bullet$ \
The simplest way to obtain a fractionally improved (G)HD is to 
take either a GHD or HD as input and compute a fractionally improved (G)HD. To this end, 
an algorithm (which we refer to as  {\sf SimpleImproveHD})
visits each node $ u $ of a given GHD or HD and computes an optimal  fractional edge cover $ \gamma_u $ for the set
$B_u$ of vertices.
This algorithm is simple and computationally inexpensive, provided that we can start off with a GHD or HD that was computed before. In our case, we simply took the HD resulting from the $\hw$-computation reported in 
Figure~\ref{fig:hw}. 
Clearly, this approach is rather naive and the dependence on a concrete HD is 
unsatisfactory. We therefore move to a more sophisticated algorithm described next.

\smallskip
$\bullet$ \
The algorithm  {\sf FracImproveHD} has as input a hypergraph $ H $ and numbers $k,k' \geq 1$, where 
$k$ is an upper bound on the $\hw$ and $k'$ the desired fractionally improved $\hw$. 
We search for an FHD ${\cal D'}$ with 
${\cal D'} = \mathit{SimpleImproveHD}({\cal D})$ for some HD ${\cal D}$ of $H$ with $\mathit{width}({\cal D})  \leq k$
and $\mathit{width}({\cal D'}) \leq k'$.
In other words, this algorithm searches for the best fractionally improved HD over all HDs of width $\leq k$. 
Hence, the result is independent of any concrete HD. 

\smallskip

The experimental results with these algorithms for computing fractionally improved HDs 
are summarized in Table~\ref{tab:fhw-improve} and Table~\ref{tab:fhw-fracimprove}. 

We have applied these algorithms to all 
hypergraphs for which $\hw \leq k$ 
with $k \in \{2,3,4,5\}$ is known from Figure~\ref{fig:hw}.
The various columns of the Tables~\ref{tab:fhw-improve} and~\ref{tab:fhw-fracimprove} are as follows: the first column (labelled $\hw$) refers to the 
(upper bound on the) $\hw$ according to Figure~\ref{fig:hw}. The next 3 columns, labelled $\geq 1$, 
$[0.5,1)$, and $[0.1,0.5)$ tell us, by how much the width can be improved (if at all) 
if we compute an FHD by one of the two algorithms. 
We thus distinguish the 3 cases if, for a hypergraph of $\hw \leq k$, 
we manage to construct an FHD of width $k-c$ for 
$c \geq 1$, $c \in [0.5,1)$, or  $c \in [0.1,0.5)$. The column with label ``no'' refers to the cases where no
improvement at all
or at least no improvement by $c \geq 0.1$
was possible. The last column counts the number of timeouts. 

For instance, in the first row of Table~\ref{tab:fhw-improve}, we see that
(with the SimpleImproveHD algorithm and starting from the HD obtained by the $\hw$-computation of  Figure~\ref{fig:hw})
out of 238 hypergraphs 
with $\hw = 2$, no improvement was possible in 172 cases. In the remaining 66 cases, 
an improvement to a width of at most $2 - 0.5$ was possible in 25 cases and an improvement to $k-c$ with 
$c \in [0.1,0.5)$ was possible in 41 cases. For the hypergraphs with $\hw = 3$ in Figure~\ref{fig:hw}, 
almost half of the hypergraphs (141 out of 310) allowed at least some improvement, in particular, 104 by 
$c \in [0.5,1)$ and 12 even by at least 1. The improvements achieved for the hypergraphs with $\hw \leq 4$
and $\hw \leq 5$ are less significant.
 
\begin{table}[t]
  \centering 
    \caption{Instances solved with \improvehd}
  \label{tab:fhw-improve}%
       \begin{tabular}{crrrrr}
        \multicolumn{1}{c}{$\hw$}  & \multicolumn{1}{c}{$\geq1$} & 
\multicolumn{1}{c}{$[0.5,1)$} & \multicolumn{1}{c}{$[0.1,0.5)$} & \multicolumn{1}{c}{no} & \multicolumn{1}{c}{timeout} \\
     \midrule

    $2$		& $0$ 		& $41$ 		& $25$ 		& $172$ 		& $0$ \\
    $3$ 	& $12$ 		& $104$ 		& $25$ 		& $169$ 	& $0$ \\
    $4$ 	& $9$ 		& $55$ 		& $11$		& $311$ 	& $0$ \\
    $5$ 	& $20$		& $14$		& $11$		& $382$ 	& $0$ \\
    $6$     & $12$      & $60$ & $80$ & $309$ & $0$ \\
     \bottomrule 
    \end{tabular}%
\end{table}%

\begin{table}[t]
  \centering 
    \caption{Instances solved with \fracimprovehd}
  \label{tab:fhw-fracimprove}%
    \begin{tabular}{crrrrr}
        \multicolumn{1}{c}{$\hw$}  & \multicolumn{1}{c}{$\geq1$} & 
\multicolumn{1}{c}{$[0.5,1)$} & \multicolumn{1}{c}{$[0.1,0.5)$} & \multicolumn{1}{c}{no} & \multicolumn{1}{c}{timeout} \\
     \midrule
    $2$     & 0   & 46    & 29  & 160  & 1  \\
     $3$    & 14  & 116   & 21  & 135  & 24 \\
    $4$     & 11  & 81    & 2  & 8  & 284  \\
    $5$     & 18  & 126    & 59  & 2  & 222  \\
    $6$     & 28   & 149    & 95  & 4  & 183  \\
     \bottomrule
    \end{tabular}%
    
\vspace{-10pt}    
    
\end{table}%

The results obtained with our implementation of the FracImproveHD algorithm are displayed in  Table \ref{tab:fhw-fracimprove}. 
We see that the number of hypergraphs which allow for a fractional improvement of the width by at least 0.5 
or even by 1 is often bigger than with SimpleImproveHD -- in particular in the cases where $k' \leq k$ with 
$k \in \{4,5\}$ holds. 
In the other cases, the results obtained with the naive SimpleImproveHD algorithm 
are not much worse than with the more sophisticated FracImproveHD algorithm. 

\section{Related Work}
\label{sect:related}
We distinguish several types of works that are highly relevant to ours. 
The works most closely related are the descriptions of HD, GHD and FHD algorithms in \cite{2002gottlob,pods/FischlGP18} and the implementation of HD computation by the 
\detkdecomp\ 
program reported in \cite{DBLP:journals/jea/GottlobS08}. 
We have extended these works in several ways. Above all, we have incorporated our analysis tool 
(reported in 
Sections \ref{sect:HyperBench} and \ref{sect:EmpiricalAnalysis}) and the GHD and FHD computations 
(reported in 
Sections \ref{sect:ghw-implementation} and \ref{sect:fractionally-improved})
into the \detkdecomp\ program -- resulting in our \newdetkdecomp\ library, which is openly available on GitHub. 
For the GHD computation, we have added heuristics to speed up the basic algorithm from \cite{pods/FischlGP18}. Moreover, we have proposed a 
novel approach via balanced separators, which allowed us to significantly extend the range of instances for which 
the GHD computation terminates in reasonable time. 
We have also introduced a new form of decomposition method: the 
fractionally improved decompositions (see Section \ref{sect:fractionally-improved}), which allow for a practical, lightweight form of FHDs.

The second important input to our work comes from the various sources 
\cite{%
arocena2015ibench%
,BenediktCQs%
,DBLP:conf/pods/BenediktKMMPST17%
,berg2017maxsat%
,geerts2014mapping%
,DBLP:journals/jea/GottlobS08%
,2015leis%
,shrjainSQLShare%
,tpch}
which we took our CQs and CSPs from. 
Note that our main goal was not to add further CQs and/or CSPs to these benchmarks. 
Instead, we have aimed at taking and combining existing, openly accessible benchmarks of CQs and CSPs, convert them into hypergraphs, which are then thoroughly analysed. Finally, the hypergraphs and the analysis results are made openly accessible again. 

The third kind of works highly relevant to ours are previous analyses of CQs and CSPs. 
To the best of our knowledge, Ghionna et al.\ \cite{DBLP:conf/icde/GhionnaGGS07} presented the first systematic study of HDs 
of benchmark CQs from TPC-H. However, Ghionna et al.\ 
pursued a research goal different  from ours in that they primarily wanted to find out to what extent HDs can actually 
speed up query evaluation. They achieved very positive results in this respect, which 
have recently been confirmed by the work of Perelman et al.\ \cite{Duncecap15a}, Tu et al. \cite{tu2015duncecap} and Aberger et al. \cite{DBLP:conf/sigmod/AbergerTOR16,DBLP:journals/tods/AbergerLTNOR17} on query evaluation using FHDs.
As a side result, Ghionna et al.\  also detected that CQs tend to have low hypertree width (a finding which was later confirmed in 
\cite{DBLP:journals/pvldb/BonifatiMT17,DBLP:conf/sigmod/PicalausaV11} and also in our study).  
In a pioneering effort, Bonifati, Martens, and Timm \cite{DBLP:journals/pvldb/BonifatiMT17} have recently
analysed an unprecedented, massive amount of  queries: they investigated 180,653,910 queries from 
(not openly available) query logs of several popular SPARQL endpoints. After elimination of duplicate queries, there were 
still 
56,164,661 queries left, out of which 26,157,880 queries were in fact CQs. 
The authors thus significantly extend previous work by 
Picalausa and Vansummeren \cite{DBLP:conf/sigmod/PicalausaV11}, who 
analysed 3,130,177 SPARQL queries posed by humans and software robots at the DBPedia SPARQL endpoint. 
The focus in \cite{DBLP:conf/sigmod/PicalausaV11} is on structural properties of SPARQL queries such as keywords used and variable structure in 
optional patterns. There is one paragraph devoted to CQs, where it is noted that 99.99\% of ca.\  2 million CQs considered 
in \cite{DBLP:conf/sigmod/PicalausaV11} are acyclic.

Many of the CQs  (over 15 million) analysed in \cite{DBLP:journals/pvldb/BonifatiMT17}
have arity 2 (here we consider the maximum arity  of all atoms in a CQ as the arity of the query), which means that all triples in such a SPARQL query have a constant at 
the predicate-position. 
Bonifati et al.\ made several interesting observations concerning the shape of these graph-like queries. For instance, they detected that
exactly one of these queries has $\tw = 3$, while all others have
$\tw \leq 2$ (and hence $\hw \leq 2$).
As far as the CQs of arity 3 are concerned (for CQs expressed as SPARQL queries, this is the maximum arity achievable), among many characteristics, also the hypertree width was computed by using the original $\detkdecomp$ program from 
\cite{DBLP:journals/jea/GottlobS08}. 
Out of 6,959,510 CQs of arity 3, 
only 86 (i.e. 0.01\textperthousand) turned out to have $\hw = 2$ and 8 queries had $¸\hw = 3$, while all other CQs of arity 3 are acyclic.  
Our analysis confirms that, also for non-random CQs of arity $> 3$, the hypertree width indeed tends to be low, with the majority of queries being even acyclic.

For the analysis of CSPs, much less work has been done. 
Although it has been shown that exploiting {(hyper-)}
tree decompositions may significantly improve the performance of CSP solving 
\cite{DBLP:journals/aicom/AmrounHA16,DBLP:journals/jetai/HabbasAS15,DBLP:conf/sara/KarakashianWC11,LalouHA09}, 
a systematic study on the (generalized) hypertree width of CSP instances has only been carried out by few works \cite{DBLP:journals/jea/GottlobS08,LalouHA09,Schafhauser06}. To the best of our knowledge, we are the first to analyse the \hw, \ghw, and \fhw\ of ca.\ 2,000 CSP instances, where most of these instances have not been studied in this respect before.

It should be noted that the focus of our work is different from the above mentioned previous works: above all, we wanted to test the 
practical feasibility of various algorithms for HD, GHD, and FHD computation (including both, previously presented algorithms and new ones developed as part of this work). As far as our repository of hypergraphs (obtained from CQs and CSPs) is concerned, we emphasize open accessibility. 
Thus, users can analyse their CQs and CSPs (with our implementations of HD, GHD, and FHD algorithms) or they can analyse 
new decomposition algorithms (with our hypergraphs, which cover quite a broad range of characteristics). 
In fact, in the recent work on FHD computation via SMT solving \cite{CP/FichteHLS18}, 
the Hyperbench benchmark has 
already been used for 
the experimental evaluation. 
In \cite{CP/FichteHLS18} a novel approach to $\fhw$ computation via an efficient encoding of the 
check-problem for FHDs to SMT (SAT modulo Theory) is presented. The tests were carried out with 2,191 hypergraphs from 
the initial version of the HyperBench. For all of these hypergraphs we have established at least some upper bound on the $\fhw$ either by our $\hw$-computation or by one of our new algorithms presented in
Sections \ref{sect:ghw-implementation} and \ref{sect:fractionally-improved}.
In contrast, the exact algorithm in \cite{CP/FichteHLS18} found FHDs only for 1.449 instances (66\%). 
In 852 cases, both our algorithms and the algorithm in \cite{CP/FichteHLS18} found FHDs of the same width; in 560 cases, an FHD of lower width was found in \cite{CP/FichteHLS18}.
By using the same benchmark for the tests,
the results  in \cite{CP/FichteHLS18} and ours 
are comparable and have thus provided 
valuable input for future improvements of the algorithms by 
combining the 
different strengths and weaknesses of the two 
approaches.

The use of the same benchmark has also allowed us to provide feedback to the authors of \cite{CP/FichteHLS18} for debugging their system: in 9 out of 2,191  cases, the ``optimal'' value for the $\fhw$ computed in [19] was apparently erroneous, since it was higher than the $\hw$ found out by our analysis; note that upper bounds on 
the width are, in general, more reliable than lower bounds since it is easy to verify if a given decomposition indeed has the desired properties, whereas ruling out the existence 
of a decomposition of a certain width is a complex and error-prone task.

\section{Conclusion}
\label{sect:conclusion}

In this work, we have presented HyperBench, a new and comprehensive benchmark of hypergraphs derived from CQs and CSPs from various areas, together with the results of extensive empirical analyses with this benchmark.

\medskip

\noindent
{\bf Lessons learned.} The empirical study has brought many insights.
Below, we summarize the most important lessons from our studies.

\smallskip
$\bullet$ \
The finding of \cite{DBLP:journals/pvldb/BonifatiMT17,DBLP:conf/sigmod/PicalausaV11} that 
non-random CQs have low hypertree width 
has been confirmed by our analysis, even if
(in contrast to SPARQL queries)
the arity of the CQs  is not bounded  by 3. 
For random CQs and CSPs, we have detected a correlation between the arity and the hypertree width, although 
also in this case, the increase of the $\hw$ with increased arity is not dramatic. 

\smallskip
$\bullet$ \
In \cite{pods/FischlGP18}, several hypergraph invariants were identified, which make the computation of GHDs and the approximation of FHDs tractable. We have seen that, at least for 
non-random instances, these invariants indeed have low values.

\smallskip
$\bullet$ \
The reduction of the $\ghw$-computation problem to the $\hw$-computation problem in case of low intersection 
width turned out to be more problematical than the theoretical trac\-tability results from \cite{pods/FischlGP18} had suggested. Even the improvement by ``local'' computation of the additional subedges did not help much.  
However, we were able to improve this significantly by presenting 
a new algorithm based on ``balanced separators''. In particular for negative instances (i.e., those with a no-answer), this approach proved very effective.

\smallskip
$\bullet$ \
An additional benefit of the new $\ghw$-algorithm based on ``balanced separators'' is that it allowed us to also fill gaps in the $\hw$-computation. Indeed, in several cases, 
we managed to verify 
$\hw \leq k$ for some $k$ but we could not show $\hw \not\leq k-1$, due to a timeout for
\rec{HD,\,$k-1$}. 
By establishing $\ghw \not\leq k-1$ with our 
new GHD-algorithm, we have implicitly showed $\hw \not\leq k-1$. This allowed us to 
compute the exact $\hw$ of many further~hypergraphs.

\smallskip
$\bullet$ \
Most surprisingly, the discrepancy between $\hw$ and $\ghw$ is much lower than expected. 
Theoretically, only the upper bound $\hw \leq 3 \cdot \ghw + 1$ is known. 
However, in practice, when considering hypergraphs of $\hw \leq 6$,
we could show that in 53\% of all cases, $\hw$ and $\ghw$ are simply identical. 
Moreover, in {\em all\/} cases when one of our implementations of $\ghw$-computation terminated on instances
with $\hw \leq 5$, we got identical values for $\hw$ and $\ghw$. 

\medskip

\noindent
{\bf Future work.} Our empirical  study has also given us many hints for future directions of research. We find the following tasks particularly urgent and/or rewarding.

\smallskip
$\bullet$ \
So far, we have only implemented the $\ghw$-computation in case of low
intersection width. In \cite{pods/FischlGP18}, tractability of the  
\rec{GHD,\,$k$} problem was also proved for the more relaxed bounded multi-intersection width. Our empirical results in 
Figure~\ref{fig:hg-props-detail} show that, apart from the random CQs and random CSPs, 
the 3-multi-intersection is $\leq 2$ in almost all cases. It seems therefore worthwhile to 
implement and test also the BMIP-algorithm from \cite{pods/FischlGP18}. 

\smallskip
$\bullet$ \
The three approaches for $\ghw$-computation presented here turned out to have complementary strengths and weaknesses. This was profitable when running all three algorithms in parallel and taking the result of the first one that terminates (see Table~\ref{tab:ghw}). In the future, we also want to implement a more sophisticated combination of the various approaches: for instance, one could try to apply our new ``balanced separator'' algorithm recursively only down to a certain recursion depth (say depth 2 or 3) to split a big given hypergraph into smaller subhypergraphs and then continue with the 
 ``global'' or ``local'' computation from 
Section~\ref{sect:ghw-implementation}.

\smallskip
$\bullet$ \
Our new approach to $\ghw$-computation via ``balanced separators'' proved quite effective in our experiments. However, further theoretical underpinning 
of this 
approach is missing. The empirical results obtained for our 
new GHD algorithm via balanced separators suggest that the
number of balanced separators is often drastically 
smaller than the number of arbitrary separators. %
We want to determine a realistic upper bound 
on the number of balanced separators 
in terms of $n$ (the number of edges) and $k$ (an upper bound on the width). This will then allow us to compute also a realistic upper bound on the runtime of this new algorithm.

\smallskip
$\bullet$ \
Finally, we want to further extend the HyperBench 
benchmark and tool in several directions. 
We will thus incorporate further implementations of decomposition algorithms from the literature 
such as the GHD- and FHD computation in \cite{moll2012} or the polynomial-time FHD computation for hypergraphs of bounded degree in \cite{FGP2017tractablefhw}. 
Moreover, we will continue to fill in hypergraphs from further sources of CSPs and CQs. For instance, in 
  \cite{DBLP:journals/tods/AbergerLTNOR17,DBLP:conf/pods/CarmeliKK17,DBLP:conf/icde/GhionnaGGS07,DBLP:conf/cikm/GhionnaGS11}  
  a collection of CQs for the experimental evaluations in those papers is mentioned. 
We will invite the authors to disclose these CQs and incorporate them into the HyperBench benchmark.

\smallskip
$\bullet$ \
Very recently, a new, huge, publically available query log has been reported in \cite{malyshevgetting}. It contains 
over 200 million SPARQL queries on Wikidata. In the paper, the anonymisation and publication of the query logs is mentioned as future work.  However, on their web site, the authors have meanwhile 
made these queries available. At first glance, these queries seem to display a similar behaviour as the SPARQL queries collected by Bonifatti et al.\ \cite{DBLP:journals/pvldb/BonifatiMT17}: there is a big number of single-atom queries and again, the vast majority of the queries is acyclic. A detailed analysis of the query log in the style of \cite{DBLP:journals/pvldb/BonifatiMT17} constitutes an important goal for future research.

\subsection*{Acknowledgements}
We would like to thank Angela Bonifati, Wim Martens, and Thomas Timm for sharing most of the hypergraphs with $\hw \geq 2$ from their work  
\cite{DBLP:journals/pvldb/BonifatiMT17} and for their effort in anonymising these hypergraphs, which was required by the license restrictions.

\clearpage
\appendix

 \noindent
{\bf \huge Appendix}

\bigskip
\noindent
In this appendix, we present further details of our analyses. First, we will present the web tool to browse and discover the hypergraphs that we have used. In addition, above all, we provide additional figures and tables to allow for a more fine-grained view on the CSP instances.
Recall from Section  \ref{sect:HyperBench} that we are dealing with 3 classes of CSP instance here: 
{\em CSP Application\/} and {\em CSP Random\/}, which are 
both taken from \url{xcsp.org} \cite{xcsp}, and  {\em CSP Other\/}, which have already been analysed w.r.t.\ $\hw$ in previous 
works \cite{DBLP:journals/jea/GottlobS08,berg2017maxsat}.
Due to lack of space, the main body of the text contains only figures and tables with aggregated values for all CSPs. 
Below, additional details 
are provided for each figure and table from the main body by distinguishing the two classes of CQs and 
three classes of CSPs.

\section{Web tool}
\label{sec:webtool}

  \begin{figure}[t] 
    \centering 
     \fbox{\includegraphics[width=0.45\textwidth]{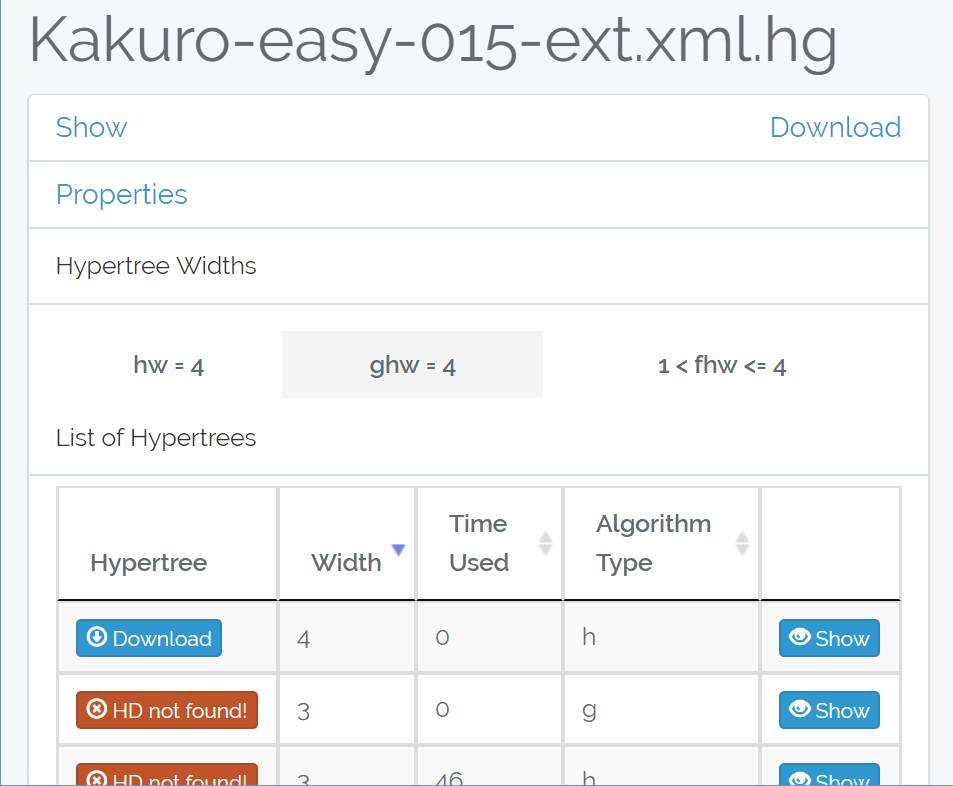}}
    \caption{HyperBench web tool: available at  \\ {\tt http://hyperbench.dbai.tuwien.ac.at}} 
    \label{fig:screenshot} 
  \end{figure}  
  
The hypergraphs in the benchmark and the results of the analyses of these
hypergraphs can be accessed via a web tool, which is available at 
\url{http://hyperbench.dbai.tuwien.ac.at}. There we have uploaded 
3,070 hypergraphs 
together with over 5,518 HDs and the output of over 16,585 further algorithm runs, where no HD of desired width was found
(either because a lower bound on the width was established
or the algorithm timed out). 
For example, in the screenshot in 
Figure~\ref{fig:screenshot} the results for the CSP instance 
``Kakuro-easy-015-ext.xml.hg'' are displayed. The hypertree width of the instance
is calculated according to the list of HDs
at the bottom of the screenshot. For this instance we have several algorithm runs, some of which led to a HD, some did not (``HD not found''). With these
we were able to pinpoint that $\hw=\ghw=4$ holds. All instances can be explored 
in such a way.

Additionally, we allow the user 
to browse, download and inspect hypergraph categories 
presented in this work. In the near future, we will also provide a search
interface to download instances having specific properties (e.g. $\hw < 5$ 
or {\it BIP} $< 3$, etc.) and to contribute to the benchmark 
by uploading hypergraphs, which are then
analysed and incorporated into our HyperBench benchmark.

\section{Further Details for Section \ref{sect:HyperBench}}
\label{app:sect-three}

Table \ref{tab:hg-sizes-detail} presents the exact numbers used in Figure \ref{fig:hg-sizes}.

\begin{table}[htbp]
  \centering 
       \begin{tabular}{rrrrrrr}
       \multicolumn{7}{c}{\bf Vertices} \\
    \toprule
         & \multicolumn{1}{c}{\rot{$1-10$}} & 
\multicolumn{1}{c}{\rot{$11-20$}} & \multicolumn{1}{c}{\rot{$21-30$}} & 
\multicolumn{1}{c}{\rot{$31-40$}} & \multicolumn{1}{c}{\rot{$41-50$}} & \multicolumn{1}{c}{\rot{$>50$}} \\
     \midrule
    {\it CQ Application}     & 199     & 126     & 67    & 51    & 29 & 63 \\
    {\it CQ Random}     & 28     & 59     & 58    & 62    & 60 & 233 \\
    {\it CSP Application}     & 2     & 204    & 97    & 133     & 27 & 627 \\
    {\it CSP Random}  & 10    & 224     & 269     & 210     & 0 & 150 \\
    {\it CSP Other}  & 0    & 3     & 2     & 0     & 7 & 70 \\ \bottomrule
    \end{tabular}%
    \;
       \begin{tabular}{rrrrrrr}
       \multicolumn{7}{c}{\bf Edges} \\
    \toprule
         & \multicolumn{1}{c}{\rot{$1-10$}} & 
\multicolumn{1}{c}{\rot{$11-20$}} & \multicolumn{1}{c}{\rot{$21-30$}} & 
\multicolumn{1}{c}{\rot{$31-40$}} & \multicolumn{1}{c}{\rot{$41-50$}} & \multicolumn{1}{c}{\rot{$>50$}} \\
     \midrule
    {\it CQ Application}     & 462     & 47     & 11    & 1    & 6 & 8 \\
    {\it CQ Random}     & 75     & 96     & 106    & 106    & 117 & 0 \\
    {\it CSP Application}     & 1     & 5    & 12    & 76     & 85 & 911 \\
    {\it CSP Random}  & 20    & 55     & 50     & 110     & 0 & 628 \\
    {\it CSP Other}  & 0   & 2     & 1     & 2     & 2 & 75 \\ \bottomrule
    \end{tabular}%
    \;
    \begin{tabular}{rrrrrr}
       \multicolumn{6}{c}{\bf Arities} \\
    \toprule
         & \multicolumn{1}{c}{\rot{$1-5$}} & 
\multicolumn{1}{c}{\rot{$6-10$}} & \multicolumn{1}{c}{\rot{$11-15$}} & 
\multicolumn{1}{c}{\rot{$16-20$}} & \multicolumn{1}{c}{\rot{$>20$}} \\
     \midrule
    {\it CQ Application}     & 315     & 64     & 70    & 41    & 45 \\
    {\it CQ Random}     & 102     & 144     & 160    & 94    & 0 \\
    {\it CSP Application}     & 667     & 263    & 1    & 56     & 103 \\
    {\it CSP Random}  & 594    & 244     & 25     & 0     & 0  \\
    {\it CSP Other}  & 74   & 7     & 1     & 0     & 0  \\ \bottomrule   
    \vspace{1pt}
    \end{tabular}%
  \caption{Hypergraph Sizes}
  \label{tab:hg-sizes-detail}%
\end{table}%

\clearpage

\section{Further Details for Section  \ref{sect:EmpiricalAnalysis}}
\label{app:sect-four}

In Table~\ref{tab:hg-props}, statistics on several hypergraph invariants were provided, namely degree, intersection width, 
$c$-multi-intersection width for $c \in \{3,4\}$, and VC-dimension. 
In Table \ref{tab:hg-props-detail}
and Figure  \ref{fig:hg-props-detail}, additional details are provided by distinguishing 
the three classes of CSP instances.

\begin{table}[htbp]
  \centering 
       \begin{tabular}{rrrrrr}
       \multicolumn{6}{c}{\it CQ Application} \\
    \toprule
        \multicolumn{1}{c}{$i$}  & \multicolumn{1}{c}{Deg} & 
\multicolumn{1}{c}{BIP} & \multicolumn{1}{c}{3-BMIP} & 
\multicolumn{1}{c}{4-BMIP} & \multicolumn{1}{c}{VC-dim} \\
     \midrule
    0     & 0     & 0     & 118   & 173   & 10 \\
    1     & 2     & 421   & 348   & 302   & 393 \\
    2     & 176   & 85    & 59    & 50    & 132 \\
    3     & 137   & 7     & 5     & 5     & 0 \\
    4     & 87    & 5     & 5     & 5     & 0 \\
    5     & 35    & 17    & 0     & 0     & 0 \\
    6     & 98    & 0     & 0     & 0     & 0 \\ \bottomrule
    \end{tabular}%
    \;
    \begin{tabular}{rrrrrr}
       \multicolumn{6}{c}{\it CQ Random} \\
    \toprule
        \multicolumn{1}{c}{$i$}  & \multicolumn{1}{c}{Deg} & 
\multicolumn{1}{c}{BIP} & \multicolumn{1}{c}{3-BMIP} & 
\multicolumn{1}{c}{4-BMIP} & \multicolumn{1}{c}{VC-dim} \\
     \midrule
    0     & 0     & 1     & 16    & 49    & 0 \\
    1     & 1     & 17    & 77    & 125   & 20 \\
    2     & 15    & 53    & 90    & 120   & 133 \\
    3     & 38    & 62    & 103   & 74    & 240 \\
    4     & 31    & 63    & 62    & 42    & 106 \\
    5     & 33    & 71    & 47    & 28    & 1 \\
    6     & 382   & 233   & 105   & 62    & 0 \\ \bottomrule
    \end{tabular}%
    \;
       \begin{tabular}{rrrrrr}
    \multicolumn{6}{c}{\it CSP Application} \\
    \toprule
        \multicolumn{1}{c}{$i$}  & \multicolumn{1}{c}{Deg} & 
\multicolumn{1}{c}{BIP} & \multicolumn{1}{c}{3-BMIP} & 
\multicolumn{1}{c}{4-BMIP} & \multicolumn{1}{c}{VC-dim} \\
     \midrule
    0     & 0     & 0     & 596   & 597   & 0 \\
    1     & 0     & 1030  & 459   & 486   & 0 \\
    2     & 596   & 59    & 34    & 7     & 1064 \\
    3     & 1     & 0     & 1     & 0     & 26 \\
    4     & 1     & 0     & 0     & 0     & 0 \\
    5     & 2     & 0     & 0     & 0     & 0 \\
    $>$5   & 490   & 1     & 0     & 0     & 0 \\
    \bottomrule
    \end{tabular}
    \;
    \begin{tabular}{rrrrrr}
    \multicolumn{6}{c}{\it CSP Random} \\
    \toprule
          \multicolumn{1}{c}{$i$}  & \multicolumn{1}{c}{Deg} & 
\multicolumn{1}{c}{BIP} & \multicolumn{1}{c}{3-BMIP} & 
\multicolumn{1}{c}{4-BMIP} & \multicolumn{1}{c}{VC-dim} \\
    \midrule
    0     & 0     & 0     & 0     & 0     & 0 \\
    1     & 0     & 200   & 200   & 238   & 0 \\
    2     & 0     & 224   & 312   & 407   & 220 \\
    3     & 0     & 76    & 147   & 95    & 515 \\
    4     & 12    & 181   & 161   & 97    & 57 \\
    5     & 8     & 99    & 14    & 1     & 71 \\
    $>$5    & 843   & 83    & 29    & 25    & 0 \\
    \bottomrule
    \end{tabular}%
    \;
    \begin{tabular}{rrrrrr}
    \multicolumn{6}{c}{\it CSP Other} \\
    \toprule
        \multicolumn{1}{c}{$i$}  & \multicolumn{1}{c}{Deg} & 
\multicolumn{1}{c}{BIP} & \multicolumn{1}{c}{3-BMIP} & 
\multicolumn{1}{c}{4-BMIP} & \multicolumn{1}{c}{VC-dim} \\
     \midrule
    0     & 0     & 0    & 1    & 6    & 0 \\
    1     & 0     & 7     & 36    & 39    & 0 \\
    2     & 1     & 36     & 23    & 16    & 51 \\
    3     & 5     & 29    & 20    & 21    & 26 \\
    4     & 19     & 10     & 2    & 0    & 0 \\
    5     & 4     & 0    & 0    & 0     & 0 \\
    $>5$     & 53   & 0     & 0     & 0     & 0 \\ \bottomrule
        \vspace{1pt}    
    \end{tabular}
      \caption{Hypergraph properties}
  \label{tab:hg-props-detail}%
\end{table}%

\begin{figure}[htbp]
\centering 
     \fbox{\includegraphics[width=0.45\textwidth]{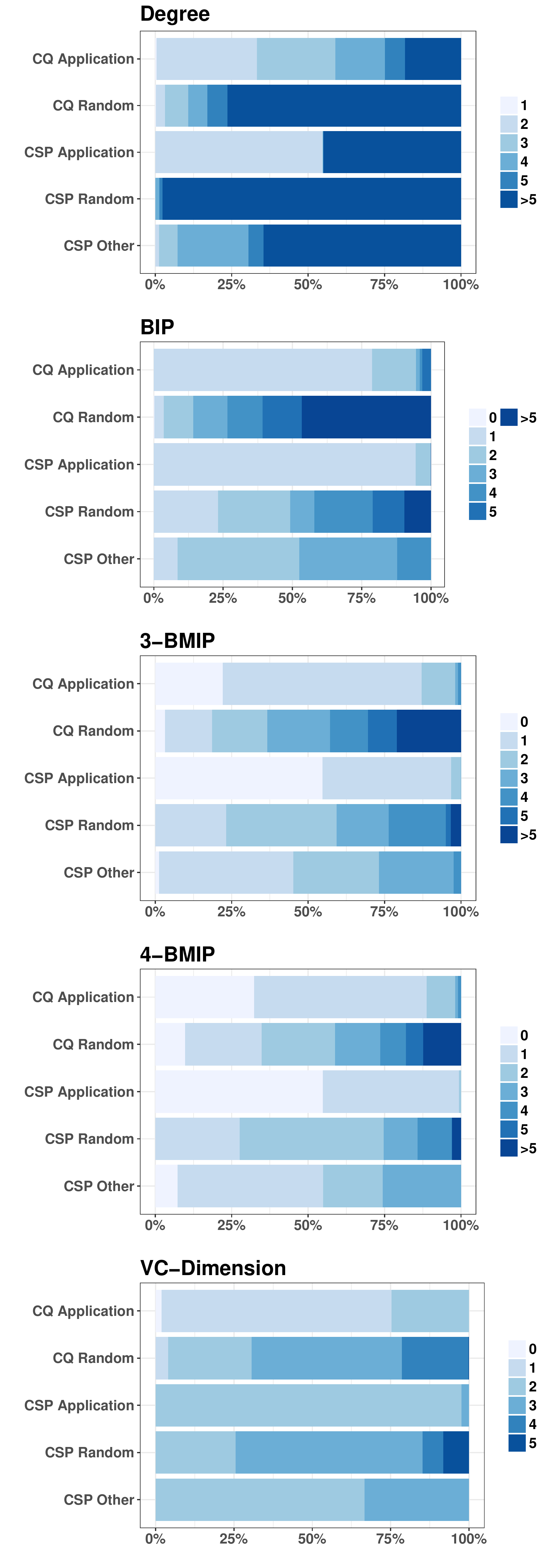}}
    \caption{Hypergraph Properties} 
\label{fig:hg-props-detail}       
\end{figure}      

\begin{table}[t]
  \centering 
       \begin{tabular}{rrrr}
       \multicolumn{4}{c}{\it CQ Application} \\
    \toprule
        \multicolumn{1}{c}{$k$}  & \multicolumn{1}{c}{yes} & 
\multicolumn{1}{c}{no} & \multicolumn{1}{c}{timeout} \\
     \midrule
    $1$     & 454 (0)    & 81 (0)    & 0   \\
     $2$     & 73 (0) &  8 (0)     & 0   \\
     $3$     & 8 (0) &  0     & 0   \\
     \bottomrule
    \end{tabular}%
    \;
    \begin{tabular}{rrrr}
       \multicolumn{4}{c}{\it CQ Random} \\
    \toprule
        \multicolumn{1}{c}{$k$}  & \multicolumn{1}{c}{yes} & 
\multicolumn{1}{c}{no} & \multicolumn{1}{c}{timeout} \\
     \midrule
    $1$     & 36 (0)    & 464 (0)    & 0   \\
     $2$     & 68 (0) & 396 (0)  &  0   \\
    $3$     & 70 (0) & 326 (32)  &  0   \\
    $4$     & 59 (0) & 167 (544)  &  100   \\
    $5$     & 54 (0) & 55 (610)  &  158   \\
    $10$     & 206 (5) & 0  &  7   \\
    $15$     & 7 (0) & 0  &  0   \\
     \bottomrule
    \end{tabular}%
    \;
       \begin{tabular}{rrrr}
       \multicolumn{4}{c}{\it CSP Application} \\
    \toprule
        \multicolumn{1}{c}{$k$}  & \multicolumn{1}{c}{yes} & 
\multicolumn{1}{c}{no} & \multicolumn{1}{c}{timeout} \\
     \midrule
    $1$     & 0    & 1090 (0)    & 0   \\
     $2$     & 29 (0) & 1061 (0)  &  0   \\
    $3$     & 116 (0) & 802 (736)  &  143   \\
    $4$     & 283 (18) & 62 (707)  &  600   \\
    $5$     & 231 (13) & 0  &  431   \\
    $10$     & 261 (0) & 0  &  170   \\
    $15$     & 12 (0) & 0  &  158   \\
     $25$     & 118 (0) & 0  &  40   \\
     $50$     & 40 (0) & 0  &  0   \\
     \bottomrule
    \end{tabular}%
    \;
    \begin{tabular}{rrrr}
       \multicolumn{4}{c}{\it CSP Random} \\
    \toprule
        \multicolumn{1}{c}{$k$}  & \multicolumn{1}{c}{yes} & 
\multicolumn{1}{c}{no} & \multicolumn{1}{c}{timeout} \\
     \midrule
    $1$     & 0    & 863 (0)    & 0   \\
     $2$     & 47 (0) & 816 (1)  &  0   \\
    $3$     & 111 (0) & 602 (1319) &  103   \\
    $4$     & 39 (42) & 160 (1332)  &  506   \\
    $5$     & 136 (59) & 0  &  530   \\
    $10$     & 530 (0) & 0  &  0   \\
     \bottomrule
    \end{tabular}%
    \; 
    \begin{tabular}{rrrr}
       \multicolumn{4}{c}{\it CSP Other} \\
    \toprule
        \multicolumn{1}{c}{$k$}  & \multicolumn{1}{c}{yes} & 
\multicolumn{1}{c}{no} & \multicolumn{1}{c}{timeout} \\
    \midrule
    $1$     & 0    & 82 (1)    & 0   \\
     $2$     & 19 (0) & 55 (219)  &  8   \\
    $3$     & 5 (0) & 11 (1257)  &  47   \\ 
    $4$     & 5 (0) & 2 (943)  &  51   \\
    $5$     & 6 (0) & 1 (0)  &  46   \\
    $10$     & 24 (0) & 0  &  23   \\
    $15$     & 6 (1) & 0  &  17   \\
    $25$     & 7 (10) & 0  &  10   \\
    $50$     & 5 (0) & 0  &  5   \\
    $75$     & 4 (0) & 0  &  1   \\
     \bottomrule
        \vspace{1pt}
    \end{tabular}%
      \caption{HW of instances with average runtime in s}
  \label{tab:hw-detail}%
\end{table}%

\noindent
In Figure~\ref{fig:hw}, the results of our $\hw$-analysis were presented. 
In Table~\ref{tab:hw-detail}
and Figure~\ref{fig:hw-detail}, additional details are provided by distinguishing 
the three classes of CSP instances. As in Figure~\ref{fig:hw}, we also  write the average runtimes in the bars of Figure  \ref{fig:hw-detail}. 
In the tabular presentation in  Table \ref{tab:hw-detail}, this information is given by putting the number of seconds in parentheses.

\begin{figure}[h]
\centering 
     \fbox{\includegraphics[width=0.45\textwidth]{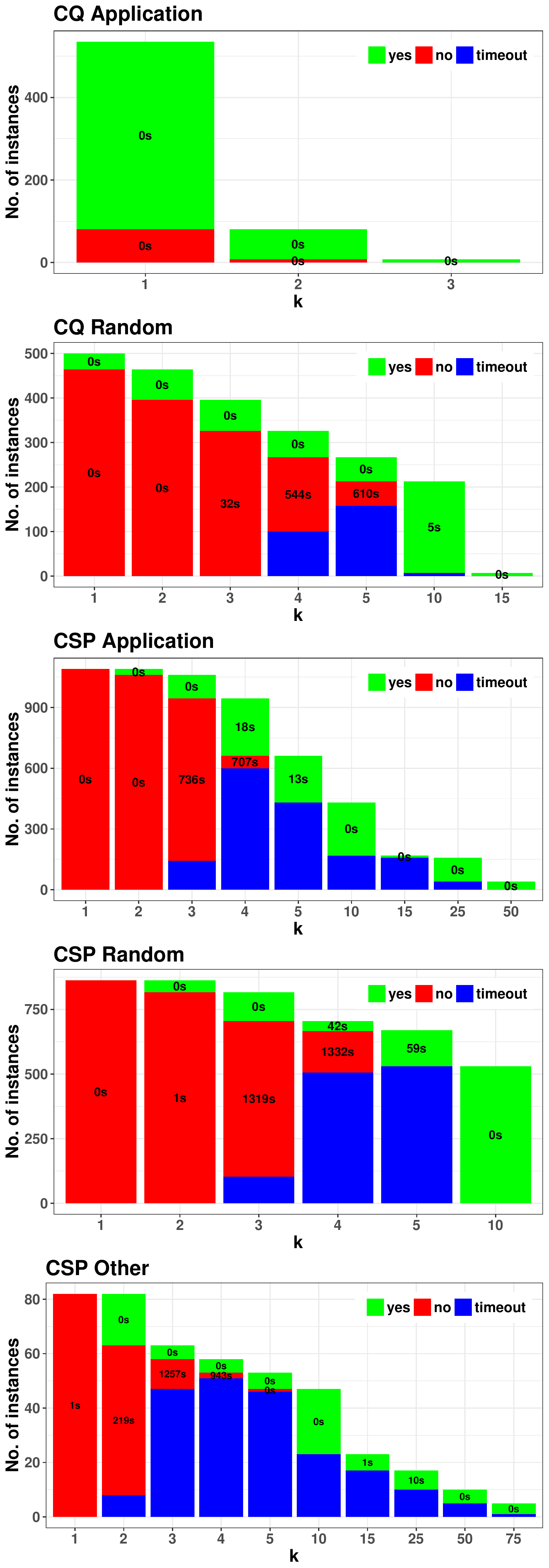}}
    \caption{HW analysis (labels are average runtimes in s)} 
\label{fig:hw-detail}      
\end{figure}   

\clearpage

\section{Further Details for Section \ref{sect:ghw-implementation}}
\label{app:sect-five}

In Table~\ref{tab:ghw}, we gave an overview of the improvements of the width when switching from $\hw$ to $\ghw$. 
In Table \ref{tab:ghw-detail}
and Figure  \ref{fig:ghw-detail}, additional details are provided by distinguishing 
the three classes of CSP instances. The runtimes are given in parentheses in the tabular representation and in 
the bars of the bar chart, respectively. The pseudocode of the GHD-algorithm via balanced separators is given in 
Figure  \ref{fig:rec-algo-detail}.

\begin{table}[htbp]
  \centering 
    \begin{tabular}{lrrr}
       \multicolumn{4}{c}{\it CQ Applicatoin} \\
    \toprule
        \multicolumn{1}{c}{$\hw \ra \ghw$}  & \multicolumn{1}{c}{yes} & 
\multicolumn{1}{c}{no} & \multicolumn{1}{c}{timeout} \\
     \midrule
    $3 \ra 2$     & 0   & 8 (0)    & 0   \\
     \bottomrule
     \multicolumn{4}{c}{\it CQ Random} \\
    \toprule
        \multicolumn{1}{c}{$\hw \ra \ghw$}  & \multicolumn{1}{c}{yes} & 
\multicolumn{1}{c}{no} & \multicolumn{1}{c}{timeout} \\
     \midrule
    $3 \ra 2$     & 0   & 70 (29)    & 0   \\
     $4 \ra 3$    & 0   & 48 (65)    & 11   \\
    $5 \ra 4$     & 0   & 30 (41)    & 24   \\
    $6 \ra 5$     & 0    & 29 (655)    & 40   \\
     \bottomrule
       \multicolumn{4}{c}{\it CSP Application} \\
    \toprule
        \multicolumn{1}{c}{$\hw \ra \ghw$}  & \multicolumn{1}{c}{yes} & 
\multicolumn{1}{c}{no} & \multicolumn{1}{c}{timeout} \\
     \midrule
    $3 \ra 2$     & 0   & 116 (7)    & 0   \\
     $4 \ra 3$    & 0   & 173 (66)    & 110   \\
    $5 \ra 4$     & 0   & 32 (9)    & 199   \\
    $6 \ra 5$     & 8 (41)   & 29 (458)    & 74   \\
     \bottomrule
       \multicolumn{4}{c}{\it CSP Random} \\
    \toprule
        \multicolumn{1}{c}{$\hw \ra \ghw$}  & \multicolumn{1}{c}{yes} & 
\multicolumn{1}{c}{no} & \multicolumn{1}{c}{timeout} \\
     \midrule
    $3 \ra 2$     & 0   & 111 (0)    & 0   \\
     $4 \ra 3$    & 0   & 38 (0)    & 1   \\
     $5 \ra 4$     & 0  & 86 (4)    & 50   \\
    $6 \ra 5$     & 9 (221)  & 121 (47)    & 141   \\
     \bottomrule
       \multicolumn{4}{c}{\it CSP Other} \\
    \toprule
        \multicolumn{1}{c}{$\hw \ra \ghw$}  & \multicolumn{1}{c}{yes} & 
\multicolumn{1}{c}{no} & \multicolumn{1}{c}{timeout} \\
    \midrule
    $3 \ra 2$     & 0   & 4 (14)    & 1   \\
     $4 \ra 3$    & 0   & 3 (120)    & 2   \\
    $5 \ra 4$     & 0   & 0   & 6   \\
    $6 \ra 5$     & 1 (2)   & 1 (2)   & 6   \\    \bottomrule
    \vspace{1pt}
    \end{tabular}%
        \caption{GHW of instances with average runtime in s}
  \label{tab:ghw-detail}%
\end{table}%

\begin{figure}[htbp]
\centering 
     \fbox{\includegraphics[width=0.45\textwidth]{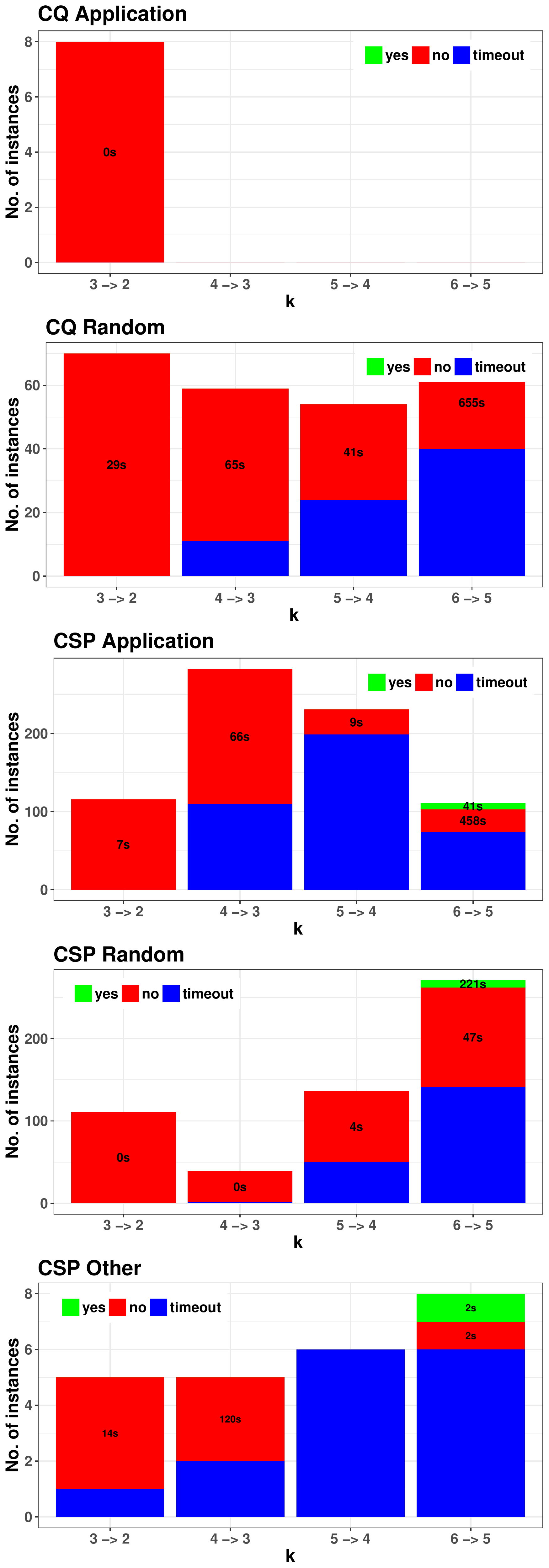}}
    \caption{GHW analysis (labels are average runtime in s)} 
\label{fig:ghw-detail}     
\end{figure}

\begin{figure}[t]
\fbox{
\begin{minipage}{\textwidth}
\begin{small}
\begin{center} 

\begin{tabbing}
xx \= xxx \= xxx \= xxx \= xxx \= xxx \= xxx \= xxx \= \kill
{\bf ALGORITHM} Find\_GHD\_via\_balancedSeparators \\
// high-level sketch  \\[1.1ex]
{\bf Input:} \hskip 7pt  hypergraph $H'$, integer $k\geq 0$. \\
{\bf Output:}  a GHD $\left<T,B_u,\lambda_u\right>$ of width $\leq k$ if exists, \\ 
\hskip 32pt   ``Reject'', otherwise.  \\[1.2ex]
 
{\bf Procedure} Find\_GHD ($H$: Hypergraph, $\Sp$: Set of Special-Edges) \\
{\bf begin}   \\
// 1.\ Stop if there are at most two special edges left: \\
\> {\bf If} $E(H) = \emptyset$ and $|\Sp| \leq 2$  \\
\> \> {\bf then} {\bf return} a GHD having a node for each $s \in \Sp$ with \\
\> \> \> \> \> $B_u := s$ and $\lambda_u := \{ s \}$;  
\\[1.1ex]
// 2.\ Find a balanced separator $\lambda_u$ for $H$: \\
\> {\bf Guess} a balanced separator $\lambda_u \subseteq E(H')$ with $|\lambda_u| \leq k$  \\
\> \> \> for root $u$ of a GHD of $H$ such that: \\
\> \> \> $\bullet \quad B(\lambda_u) \not\in \Sp$ \\
\> \> \> $\bullet \quad B(\lambda_u) \subseteq V(H)$
\\[1.1ex]
\> {\bf If} no such balanced separator exists \\
\> \> {\bf then} {\bf return} Reject; \\ 
\> $B_u := B(\lambda_u)$; 
\\[1.1ex]                
// 3.\ Split $H$ into connected components $C_1, \dots, C_\ell$ w.r.t.\ $\lambda_u$: \+ \\
$V_u := V(H) \setminus B_u$; \\
$E_u := \{ e  \cap V_u  \mid e \in (E(H)\cup \Sp) \}$; \\
 {\bf Compute} the connected components of $(V_u, E_u)$ \\
{\bf Let} the connected components be denoted by $C_1, \dots, C_\ell$; \- \\[1.1ex]                   
// 4.\ Build the pairs $\left<H_i,\Sp_i\right>$ for each connected component $C_i$: \+ \\
 {\bf For every}  $i \in \{1, \dots, \ell\}$ {\bf do}\+ \\
 $E_i := \{e \mid e\in E(H)$ and $e \cap C_i \neq \emptyset\}$; \\ %
  $V_i := V(E_i) \cup B_u$; \\
{\bf Let} $H_i$ be the hypergraph $(V_i,E_i)$; \\
 $\Sp_i := \{ s \in \Sp \mid s \cap C_i \neq \emptyset\} \cup \{B_u\}$; \- \\
{\bf od}; \- \\[1.1ex]                   

// 5.\ Call Find\_GHD($H_i$, $\Sp_i$) for each pair  $\left<H_i,\Sp_i\right>$: \+ \\
{\bf  For each} $i \in \{1, \dots, \ell\}$ $T_i :=$ Find\_GHD ($H_i$, $\Sp_i$) \\[1.1ex]

{\bf  If } $\exists i$ s.t. recursive call returns Reject \\
 \> {\bf then} {\bf  return} Reject; \- \\[1.1ex]
 
// 6.\ Create and return a new GHD for $H$ having $B_u$ and $\lambda_u$ as root: \+ \\
{\bf Create} a new root $u$ with $B_u$ and $\lambda_u$; \\
{\bf Reroot} all $T_i$ at the node $t_i$ where $\lambda_{t_i} = \{ B_u \}$; \\
{\bf Attach} all children of $t_i$ to $u$; \\
{\bf return} the new GHD rooted at $u$; \- \\
{\bf end}  \\[1.1ex]
{\bf begin} (* Main *)  \\
  \> {\bf  return} Find\_GHD ($H'$, $\emptyset$); \\
{\bf end}  

\end{tabbing}
\end{center}
\end{small}
\end{minipage}
}

\vspace{-1.8ex}

\caption{Recursive GHD-algorithm via balanced separators}
 \label{fig:rec-algo-detail}
 \end{figure}

\clearpage

\section{Further Details for Section  \ref{sect:fractionally-improved}}
\label{app:sect-six}

In Table~\ref{tab:fhw-improve}, we presented the achieved improvements of the width by switching from integral covers of HDs to fractional covers. 
In Tables~\ref{tab:fhw-simple-detail} and ~\ref{tab:fhw-frac-detail} 
and Figures~\ref{fig:fhw-simple-detail} and~\ref{fig:fhw-frac-detail},
additional details are provided by distinguishing 
the three classes of CSP instances and by showing the results of the two algorithms
SimpleImproveHD (in Table~\ref{tab:fhw-simple-detail} and Figure~\ref{fig:fhw-simple-detail})
and FracImproveHD (in Table~\ref{tab:fhw-frac-detail} and Figure~\ref{fig:fhw-frac-detail})
separately.

\begin{table}[htbp]
  \centering 
    \begin{tabular}{crrrrr}
       \multicolumn{6}{c}{\it CQ Application} \\
    \toprule
        \multicolumn{1}{c}{$\hw$}  & \multicolumn{1}{c}{$\geq 1$} & \multicolumn{1}{c}{$\geq 0.5$} & \multicolumn{1}{c}{$\geq 0.1$} &
\multicolumn{1}{c}{no} & \multicolumn{1}{c}{timeout} \\
     \midrule
    $2$     & 0   & 17    & 0 & 56 & 0   \\
    $3$     & 0   & 0    & 0 & 8 & 0   \\
     \bottomrule
     \multicolumn{6}{c}{\it CQ Random} \\
    \toprule
        \multicolumn{1}{c}{$\hw$}  & \multicolumn{1}{c}{$\geq 1$} & \multicolumn{1}{c}{$\geq 0.5$} & \multicolumn{1}{c}{$\geq 0.1$} &
\multicolumn{1}{c}{no} & \multicolumn{1}{c}{timeout} \\
     \midrule
   2&0&24&7&37&0\\
3&6&18&10&36&0\\
4&8&19&2&30&0\\
5&14&8&5&27&0\\
6&12&15&9&35&0\\
     \bottomrule
       \multicolumn{6}{c}{\it CSP Application} \\
    \toprule
        \multicolumn{1}{c}{$\hw$}  & \multicolumn{1}{c}{$\geq 1$} & \multicolumn{1}{c}{$\geq 0.5$} & \multicolumn{1}{c}{$\geq 0.1$} &
\multicolumn{1}{c}{no} & \multicolumn{1}{c}{timeout} \\
     \midrule
     $2$     & 0   & 0    & 0 & 29 & 0   \\
    $3$     & 0   & 0    & 0 & 116 & 0   \\
     $4$    & 0   & 7    & 0 & 276 & 0   \\
    $5$     & 0   & 5   & 0 & 226 & 0   \\
    $6$     & 0   & 6    & 0 & 105 & 0   \\
     \bottomrule
       \multicolumn{6}{c}{\it CSP Random} \\
    \toprule
        \multicolumn{1}{c}{$\hw$}  & \multicolumn{1}{c}{$\geq 1$} & \multicolumn{1}{c}{$\geq 0.5$} & \multicolumn{1}{c}{$\geq 0.1$} &
\multicolumn{1}{c}{no} & \multicolumn{1}{c}{timeout} \\
     \midrule
   $2$     & 0   & 0    & 18 & 29 & 0   \\
    $3$     & 6   & 86    & 15 & 4 & 0   \\
     $4$    & 1   & 29    & 9 & 0 & 0   \\
    $5$     & 6   & 1   & 6 & 123 & 0   \\
    $6$     & 0   & 39    & 70 & 162 & 0   \\
     \bottomrule
       \multicolumn{6}{c}{\it CSP Other} \\
    \toprule
        \multicolumn{1}{c}{$\hw$}  & \multicolumn{1}{c}{$\geq 1$} & \multicolumn{1}{c}{$\geq 0.5$} & \multicolumn{1}{c}{$\geq 0.1$} &
\multicolumn{1}{c}{no} & \multicolumn{1}{c}{timeout} \\
    \midrule
   $2$     & 0   & 0    & 0 & 20 & 0   \\
    $3$     & 0   & 0    & 0 & 5 & 0   \\
     $4$    & 0   & 0    & 0 & 5 & 0   \\
    $5$     & 0   & 0   & 0 & 6 & 0   \\ 
    $6$     & 0   & 0    & 1 & 7 & 0   \\
     \bottomrule
     \vspace{1pt}
    \end{tabular}%
        \caption{SimpleImproveHD of instances}
  \label{tab:fhw-simple-detail}%
\end{table}%

\begin{figure}[htbp]
\centering 
     \fbox{\includegraphics[width=0.45\textwidth]{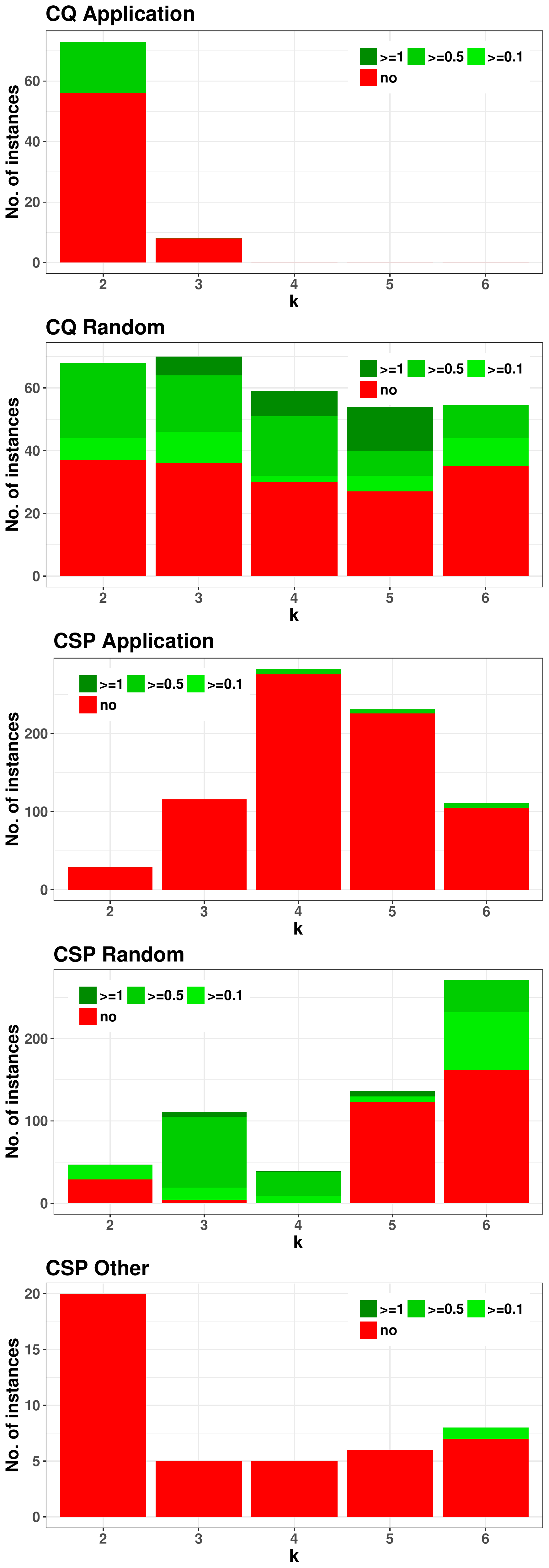}}
    \caption{SimpleImproveHD analysis} 
\label{fig:fhw-simple-detail}     
\end{figure}  

\clearpage

\begin{table}[htbp]
  \centering 
    \begin{tabular}{crrrrr}
       \multicolumn{6}{c}{\it CQ Application} \\
    \toprule
        \multicolumn{1}{c}{$\hw$}  & \multicolumn{1}{c}{$\geq 1$} & \multicolumn{1}{c}{$\geq 0.5$} & \multicolumn{1}{c}{$\geq 0.1$} &
\multicolumn{1}{c}{no} & \multicolumn{1}{c}{timeout} \\
     \midrule
    $2$     & 0   & 20    & 0 & 53 & 0   \\
    $3$     & 0   & 0    & 0 & 8 & 0   \\
     \bottomrule
     \multicolumn{6}{c}{\it CQ Random} \\
    \toprule
        \multicolumn{1}{c}{$\hw$}  & \multicolumn{1}{c}{$\geq 1$} & \multicolumn{1}{c}{$\geq 0.5$} & \multicolumn{1}{c}{$\geq 0.1$} &
\multicolumn{1}{c}{no} & \multicolumn{1}{c}{timeout} \\
     \midrule
   $2$     & 0&26&7&35&0   \\
    $3$     & 7 & 29 & 5 & 6& 23   \\
     $4$    & 10 & 23 & 1 & 6 & 19   \\
    $5$     & 12 & 19 & 4 & 1 & 18   \\
    $6$     & 14 & 26 & 7 & 0 & 22   \\
     \bottomrule 
       \multicolumn{6}{c}{\it CSP Application} \\
    \toprule
        \multicolumn{1}{c}{$\hw$}  & \multicolumn{1}{c}{$\geq 1$} & \multicolumn{1}{c}{$\geq 0.5$} & \multicolumn{1}{c}{$\geq 0.1$} &
\multicolumn{1}{c}{no} & \multicolumn{1}{c}{timeout} \\
     \midrule
     $2$     & 0   & 0    & 0 & 29 & 0   \\
    $3$     & 0   & 0    & 0 & 116 & 0   \\
     $4$    & 0   & 21    & 0 & 2 & 260   \\
    $5$     & 0   & 83   & 0 & 1 & 147   \\
    $6$     & 1   & 28    & 0 & 2 & 80   \\
     \bottomrule
       \multicolumn{6}{c}{\it CSP Random} \\
    \toprule
        \multicolumn{1}{c}{$\hw$}  & \multicolumn{1}{c}{$\geq 1$} & \multicolumn{1}{c}{$\geq 0.5$} & \multicolumn{1}{c}{$\geq 0.1$} &
\multicolumn{1}{c}{no} & \multicolumn{1}{c}{timeout} \\
     \midrule
   $2$     & 0   & 0    & 22 & 25 & 0   \\
    $3$     & 7   & 87    & 16 & 1 & 0   \\
     $4$    & 1   & 37    & 1 & 0 & 0   \\
    $5$     & 6   & 24   & 55 & 0 & 51   \\
    $6$     & 12   & 94    & 87 & 3 & 75   \\
     \bottomrule
       \multicolumn{6}{c}{\it CSP Other} \\
    \toprule
        \multicolumn{1}{c}{$\hw$}  & \multicolumn{1}{c}{$\geq 1$} & \multicolumn{1}{c}{$\geq 0.5$} & \multicolumn{1}{c}{$\geq 0.1$} &
\multicolumn{1}{c}{no} & \multicolumn{1}{c}{timeout} \\
    \midrule
   $2$     & 0   & 0    & 0 & 18 & 1   \\
    $3$     & 0   & 0    & 0 & 4 & 1   \\
     $4$    & 0   & 0    & 0 & 1 & 4   \\
    $5$     & 0   & 0   & 0 & 1 & 5   \\ 
    $6$     & 1   & 1    & 1 & 1 & 4   \\
     \bottomrule
     \vspace{1pt}
    \end{tabular}%
        \caption{FracImproveHD of instances} 
  \label{tab:fhw-frac-detail}%
\end{table}%

\begin{figure}[htbp]
\centering 
     \fbox{\includegraphics[width=0.45\textwidth]{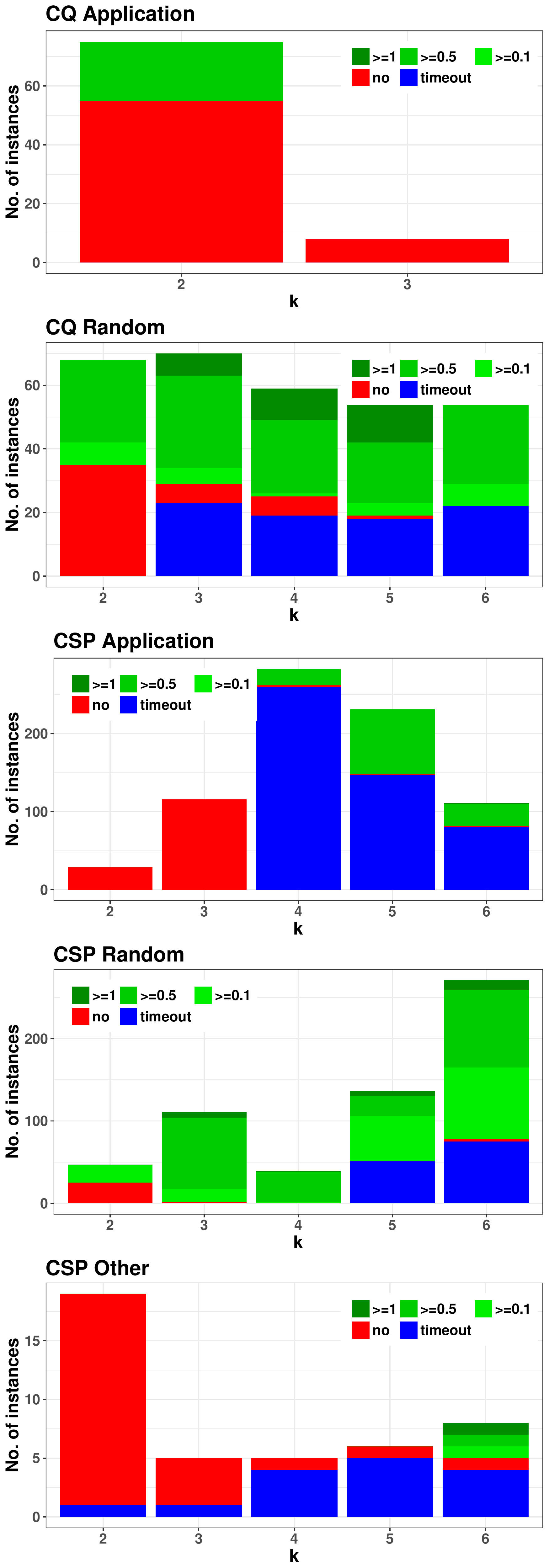}}
    \caption{FracImproveHD analysis} 
\label{fig:fhw-frac-detail}     
\end{figure}


\begin{thebibliography}{51}

%
%
%
%
%
%
%
%
%
%
%
%
%
%
%
%

\ifx \showCODEN    \undefined \def \showCODEN     #1{\unskip}     \fi
\ifx \showDOI      \undefined \def \showDOI       #1{#1}\fi
\ifx \showISBNx    \undefined \def \showISBNx     #1{\unskip}     \fi
\ifx \showISBNxiii \undefined \def \showISBNxiii  #1{\unskip}     \fi
\ifx \showISSN     \undefined \def \showISSN      #1{\unskip}     \fi
\ifx \showLCCN     \undefined \def \showLCCN      #1{\unskip}     \fi
\ifx \shownote     \undefined \def \shownote      #1{#1}          \fi
\ifx \showarticletitle \undefined \def \showarticletitle #1{#1}   \fi
\ifx \showURL      \undefined \def \showURL       {\relax}        \fi
%
%
\providecommand\bibfield[2]{#2}
\providecommand\bibinfo[2]{#2}
\providecommand\natexlab[1]{#1}
\providecommand\showeprint[2][]{arXiv:#2}

\bibitem[\protect\citeauthoryear{Aberger, Lamb, Tu, N{\"{o}}tzli, Olukotun, and
  R{\'{e}}}{Aberger et~al\mbox{.}}{2017}]%
        {DBLP:journals/tods/AbergerLTNOR17}
\bibfield{author}{\bibinfo{person}{Christopher~R. Aberger},
  \bibinfo{person}{Andrew Lamb}, \bibinfo{person}{Susan Tu},
  \bibinfo{person}{Andres N{\"{o}}tzli}, \bibinfo{person}{Kunle Olukotun},
  {and} \bibinfo{person}{Christopher R{\'{e}}}.}
  \bibinfo{year}{2017}\natexlab{}.
\newblock \showarticletitle{EmptyHeaded: {A} Relational Engine for Graph
  Processing}.
\newblock \bibinfo{journal}{\emph{{ACM} Trans. Database Syst.}}
  \bibinfo{volume}{42}, \bibinfo{number}{4} (\bibinfo{year}{2017}),
  \bibinfo{pages}{20:1--20:44}.
\newblock


\bibitem[\protect\citeauthoryear{Aberger, Tu, Olukotun, and R{\'{e}}}{Aberger
  et~al\mbox{.}}{2016a}]%
        {DBLP:conf/sigmod/AbergerTOR16}
\bibfield{author}{\bibinfo{person}{Christopher~R. Aberger},
  \bibinfo{person}{Susan Tu}, \bibinfo{person}{Kunle Olukotun}, {and}
  \bibinfo{person}{Christopher R{\'{e}}}.} \bibinfo{year}{2016}\natexlab{a}.
\newblock \showarticletitle{EmptyHeaded: {A} Relational Engine for Graph
  Processing}. In \bibinfo{booktitle}{\emph{Proc.\ {SIGMOD} 2016}}.
  \bibinfo{publisher}{{ACM}}, \bibinfo{pages}{431--446}.
\newblock


\bibitem[\protect\citeauthoryear{Aberger, Tu, Olukotun, and R{\'{e}}}{Aberger
  et~al\mbox{.}}{2016b}]%
        {aberger2016old}
\bibfield{author}{\bibinfo{person}{Christopher~R. Aberger},
  \bibinfo{person}{Susan Tu}, \bibinfo{person}{Kunle Olukotun}, {and}
  \bibinfo{person}{Christopher R{\'{e}}}.} \bibinfo{year}{2016}\natexlab{b}.
\newblock \showarticletitle{Old Techniques for New Join Algorithms: {A} Case
  Study in {RDF} Processing}.
\newblock \bibinfo{journal}{\emph{CoRR}}  \bibinfo{volume}{abs/1602.03557}
  (\bibinfo{year}{2016}).
\newblock
\urldef\tempurl%
\url{http://arxiv.org/abs/1602.03557}
\showURL{%
\tempurl}


\bibitem[\protect\citeauthoryear{Adler, Gottlob, and Grohe}{Adler
  et~al\mbox{.}}{2007}]%
        {DBLP:journals/ejc/AdlerGG07}
\bibfield{author}{\bibinfo{person}{Isolde Adler}, \bibinfo{person}{Georg
  Gottlob}, {and} \bibinfo{person}{Martin Grohe}.}
  \bibinfo{year}{2007}\natexlab{}.
\newblock \showarticletitle{Hypertree width and related hypergraph invariants}.
\newblock \bibinfo{journal}{\emph{Eur. J. Comb.}} \bibinfo{volume}{28},
  \bibinfo{number}{8} (\bibinfo{year}{2007}), \bibinfo{pages}{2167--2181}.
\newblock


\bibitem[\protect\citeauthoryear{Amroun, Habbas, and Aggoune{-}Mtalaa}{Amroun
  et~al\mbox{.}}{2016}]%
        {DBLP:journals/aicom/AmrounHA16}
\bibfield{author}{\bibinfo{person}{Kamal Amroun}, \bibinfo{person}{Zineb
  Habbas}, {and} \bibinfo{person}{Wassila Aggoune{-}Mtalaa}.}
  \bibinfo{year}{2016}\natexlab{}.
\newblock \showarticletitle{A compressed Generalized Hypertree
  Decomposition-based solving technique for non-binary Constraint Satisfaction
  Problems}.
\newblock \bibinfo{journal}{\emph{{AI} Commun.}} \bibinfo{volume}{29},
  \bibinfo{number}{2} (\bibinfo{year}{2016}), \bibinfo{pages}{371--392}.
\newblock


\bibitem[\protect\citeauthoryear{Aref, ten Cate, Green, Kimelfeld, Olteanu,
  Pasalic, Veldhuizen, and Washburn}{Aref et~al\mbox{.}}{2015}]%
        {DBLP:conf/sigmod/ArefCGKOPVW15}
\bibfield{author}{\bibinfo{person}{Molham Aref}, \bibinfo{person}{Balder ten
  Cate}, \bibinfo{person}{Todd~J. Green}, \bibinfo{person}{Benny Kimelfeld},
  \bibinfo{person}{Dan Olteanu}, \bibinfo{person}{Emir Pasalic},
  \bibinfo{person}{Todd~L. Veldhuizen}, {and} \bibinfo{person}{Geoffrey
  Washburn}.} \bibinfo{year}{2015}\natexlab{}.
\newblock \showarticletitle{Design and Implementation of the {LogicBlox}
  System}. In \bibinfo{booktitle}{\emph{Proc.\ {SIGMOD} 2015}}.
  \bibinfo{publisher}{{ACM}}.
\newblock


\bibitem[\protect\citeauthoryear{Arocena, Glavic, Ciucanu, and Miller}{Arocena
  et~al\mbox{.}}{2015}]%
        {arocena2015ibench}
\bibfield{author}{\bibinfo{person}{Patricia~C. Arocena}, \bibinfo{person}{Boris
  Glavic}, \bibinfo{person}{Radu Ciucanu}, {and} \bibinfo{person}{Ren{\'e}e~J.
  Miller}.} \bibinfo{year}{2015}\natexlab{}.
\newblock \showarticletitle{The iBench Integration Metadata Generator}.
\newblock \bibinfo{journal}{\emph{Proc. VLDB Endow.}} \bibinfo{volume}{9},
  \bibinfo{number}{3} (\bibinfo{date}{Nov.} \bibinfo{year}{2015}),
  \bibinfo{pages}{108--119}.
\newblock
\showISSN{2150-8097}


\bibitem[\protect\citeauthoryear{Atserias, Grohe, and Marx}{Atserias
  et~al\mbox{.}}{2013}]%
        {DBLP:journals/siamcomp/AtseriasGM13}
\bibfield{author}{\bibinfo{person}{Albert Atserias}, \bibinfo{person}{Martin
  Grohe}, {and} \bibinfo{person}{D{\'{a}}niel Marx}.}
  \bibinfo{year}{2013}\natexlab{}.
\newblock \showarticletitle{Size Bounds and Query Plans for Relational Joins}.
\newblock \bibinfo{journal}{\emph{{SIAM} J. Comput.}} \bibinfo{volume}{42},
  \bibinfo{number}{4} (\bibinfo{year}{2013}), \bibinfo{pages}{1737--1767}.
\newblock


\bibitem[\protect\citeauthoryear{Audemard, Boussemart, Lecoutre, and
  Piette}{Audemard et~al\mbox{.}}{2016}]%
        {xcsp}
\bibfield{author}{\bibinfo{person}{Gilles Audemard},
  \bibinfo{person}{Fr\'ed\'eric Boussemart}, \bibinfo{person}{Christoph
  Lecoutre}, {and} \bibinfo{person}{C\'edric Piette}.}
  \bibinfo{year}{2016}\natexlab{}.
\newblock \bibinfo{title}{{XCSP3}: an {XML}-based format designed to represent
  combinatorial constrained problems}.
\newblock \bibinfo{howpublished}{\url{http://xcsp.org}}.
  (\bibinfo{year}{2016}).
\newblock


\bibitem[\protect\citeauthoryear{Bakibayev, Kocisk{\'{y}}, Olteanu, and
  Z{\'{a}}vodn{\'{y}}}{Bakibayev et~al\mbox{.}}{2013}]%
        {BKOZ13}
\bibfield{author}{\bibinfo{person}{Nurzhan Bakibayev},
  \bibinfo{person}{Tom{\'{a}}s Kocisk{\'{y}}}, \bibinfo{person}{Dan Olteanu},
  {and} \bibinfo{person}{Jakub Z{\'{a}}vodn{\'{y}}}.}
  \bibinfo{year}{2013}\natexlab{}.
\newblock \showarticletitle{Aggregation and Ordering in Factorised Databases}.
\newblock \bibinfo{journal}{\emph{{PVLDB}}} \bibinfo{volume}{6},
  \bibinfo{number}{14} (\bibinfo{year}{2013}).
\newblock


\bibitem[\protect\citeauthoryear{Benedikt}{Benedikt}{2017}]%
        {BenediktCQs}
\bibfield{author}{\bibinfo{person}{Michael Benedikt}.}
  \bibinfo{year}{2017}\natexlab{}.
\newblock \bibinfo{title}{{CQ} benchmarks}.
\newblock   (\bibinfo{year}{2017}).
\newblock
\newblock
\shownote{Personal Communication.}


\bibitem[\protect\citeauthoryear{Benedikt, Konstantinidis, Mecca, Motik,
  Papotti, Santoro, and Tsamoura}{Benedikt et~al\mbox{.}}{2017}]%
        {DBLP:conf/pods/BenediktKMMPST17}
\bibfield{author}{\bibinfo{person}{Michael Benedikt}, \bibinfo{person}{George
  Konstantinidis}, \bibinfo{person}{Giansalvatore Mecca},
  \bibinfo{person}{Boris Motik}, \bibinfo{person}{Paolo Papotti},
  \bibinfo{person}{Donatello Santoro}, {and} \bibinfo{person}{Efthymia
  Tsamoura}.} \bibinfo{year}{2017}\natexlab{}.
\newblock \showarticletitle{Benchmarking the Chase}. In
  \bibinfo{booktitle}{\emph{Proc.\ {PODS} 2017}}. \bibinfo{publisher}{{ACM}},
  \bibinfo{pages}{37--52}.
\newblock


\bibitem[\protect\citeauthoryear{Berg, Lodha, J{\"a}rvisalo, and Szeider}{Berg
  et~al\mbox{.}}{2017}]%
        {berg2017maxsat}
\bibfield{author}{\bibinfo{person}{Jeremias Berg}, \bibinfo{person}{Neha
  Lodha}, \bibinfo{person}{Matti J{\"a}rvisalo}, {and} \bibinfo{person}{Stefan
  Szeider}.} \bibinfo{year}{2017}\natexlab{}.
\newblock \showarticletitle{MaxSAT Benchmarks based on Determining Generalized
  Hypertree-width}.
\newblock \bibinfo{journal}{\emph{MaxSAT Evaluation 2017}}
  (\bibinfo{year}{2017}), \bibinfo{pages}{22}.
\newblock


\bibitem[\protect\citeauthoryear{Bonifati, Martens, and Timm}{Bonifati
  et~al\mbox{.}}{2017}]%
        {DBLP:journals/pvldb/BonifatiMT17}
\bibfield{author}{\bibinfo{person}{Angela Bonifati}, \bibinfo{person}{Wim
  Martens}, {and} \bibinfo{person}{Thomas Timm}.}
  \bibinfo{year}{2017}\natexlab{}.
\newblock \showarticletitle{An Analytical Study of Large {SPARQL} Query Logs}.
\newblock \bibinfo{journal}{\emph{{PVLDB}}} \bibinfo{volume}{11},
  \bibinfo{number}{2} (\bibinfo{year}{2017}), \bibinfo{pages}{149--161}.
\newblock
\urldef\tempurl%
\url{http://www.vldb.org/pvldb/vol11/p149-bonifati.pdf}
\showURL{%
\tempurl}


\bibitem[\protect\citeauthoryear{Carmeli, Kenig, and Kimelfeld}{Carmeli
  et~al\mbox{.}}{2017}]%
        {DBLP:conf/pods/CarmeliKK17}
\bibfield{author}{\bibinfo{person}{Nofar Carmeli}, \bibinfo{person}{Batya
  Kenig}, {and} \bibinfo{person}{Benny Kimelfeld}.}
  \bibinfo{year}{2017}\natexlab{}.
\newblock \showarticletitle{Efficiently Enumerating Minimal Triangulations}. In
  \bibinfo{booktitle}{\emph{Proc.\ {PODS} 2017}}. \bibinfo{publisher}{{ACM}},
  \bibinfo{pages}{273--287}.
\newblock


\bibitem[\protect\citeauthoryear{Chandra and Merlin}{Chandra and
  Merlin}{1977}]%
        {DBLP:conf/stoc/ChandraM77}
\bibfield{author}{\bibinfo{person}{Ashok~K. Chandra} {and}
  \bibinfo{person}{Philip~M. Merlin}.} \bibinfo{year}{1977}\natexlab{}.
\newblock \showarticletitle{Optimal Implementation of Conjunctive Queries in
  Relational Data Bases}. In \bibinfo{booktitle}{\emph{Proc.\ STOC 1977}}.
  \bibinfo{publisher}{{ACM}}, \bibinfo{pages}{77--90}.
\newblock


\bibitem[\protect\citeauthoryear{Dechter}{Dechter}{2003}]%
        {dechter2003}
\bibfield{author}{\bibinfo{person}{Rina Dechter}.}
  \bibinfo{year}{2003}\natexlab{}.
\newblock \bibinfo{booktitle}{\emph{Constraint Processing}}.
\newblock


\bibitem[\protect\citeauthoryear{Feige and Mahdian}{Feige and Mahdian}{2006}]%
        {DBLP:conf/stoc/FeigeM06}
\bibfield{author}{\bibinfo{person}{Uriel Feige} {and} \bibinfo{person}{Mohammad
  Mahdian}.} \bibinfo{year}{2006}\natexlab{}.
\newblock \showarticletitle{Finding small balanced separators}. In
  \bibinfo{booktitle}{\emph{Proc.\ STOC 2006}}. \bibinfo{publisher}{{ACM}},
  \bibinfo{pages}{375--384}.
\newblock
\urldef\tempurl%
\url{https://doi.org/10.1145/1132516.1132573}
\showDOI{\tempurl}


\bibitem[\protect\citeauthoryear{Fichte, Hecher, Lodha, and Szeider}{Fichte
  et~al\mbox{.}}{2018}]%
        {CP/FichteHLS18}
\bibfield{author}{\bibinfo{person}{Johannes~K. Fichte}, \bibinfo{person}{Markus
  Hecher}, \bibinfo{person}{Neha Lodha}, {and} \bibinfo{person}{Stefan
  Szeider}.} \bibinfo{year}{2018}\natexlab{}.
\newblock \showarticletitle{An SMT Approach to Fractional Hypertree Width}. In
  \bibinfo{booktitle}{\emph{Proc.\ CP 2018}} \emph{(\bibinfo{series}{LNCS})},
  Vol.~\bibinfo{volume}{11008}. \bibinfo{publisher}{Springer},
  \bibinfo{pages}{109--127}.
\newblock


\bibitem[\protect\citeauthoryear{Fischl, Gottlob, and Pichler}{Fischl
  et~al\mbox{.}}{[n. d.]}]%
        {pods/FischlGP18}
\bibfield{author}{\bibinfo{person}{Wolfgang Fischl}, \bibinfo{person}{Georg
  Gottlob}, {and} \bibinfo{person}{Reinhard Pichler}.} \bibinfo{year}{[n.
  d.]}\natexlab{}.
\newblock \showarticletitle{General and Fractional Hypertree Decompositions:
  Hard and Easy Cases}. In \bibinfo{booktitle}{\emph{Proc.\ PODS 2018}}.
\newblock


\bibitem[\protect\citeauthoryear{Fischl, Gottlob, and Pichler}{Fischl
  et~al\mbox{.}}{2017}]%
        {FGP2017tractablefhw}
\bibfield{author}{\bibinfo{person}{Wolfgang Fischl}, \bibinfo{person}{Georg
  Gottlob}, {and} \bibinfo{person}{Reinhard Pichler}.}
  \bibinfo{year}{2017}\natexlab{}.
\newblock \showarticletitle{Tractable Cases for Recognizing Low Fractional
  Hypertree Width}.
\newblock \bibinfo{journal}{\emph{viXra.org e-prints}}
  \bibinfo{volume}{viXra:1708.0373} (\bibinfo{year}{2017}).
\newblock
\urldef\tempurl%
\url{http://vixra.org/abs/1708.0373}
\showURL{%
\tempurl}


\bibitem[\protect\citeauthoryear{Geerts, Mecca, Papotti, and Santoro}{Geerts
  et~al\mbox{.}}{2014}]%
        {geerts2014mapping}
\bibfield{author}{\bibinfo{person}{Floris Geerts},
  \bibinfo{person}{Giansalvatore Mecca}, \bibinfo{person}{Paolo Papotti}, {and}
  \bibinfo{person}{Donatello Santoro}.} \bibinfo{year}{2014}\natexlab{}.
\newblock \showarticletitle{Mapping and cleaning}. In
  \bibinfo{booktitle}{\emph{Proc.\ ICDE 2014}}. IEEE,
  \bibinfo{pages}{232--243}.
\newblock


\bibitem[\protect\citeauthoryear{Ghionna, Granata, Greco, and
  Scarcello}{Ghionna et~al\mbox{.}}{2007}]%
        {DBLP:conf/icde/GhionnaGGS07}
\bibfield{author}{\bibinfo{person}{Lucantonio Ghionna}, \bibinfo{person}{Luigi
  Granata}, \bibinfo{person}{Gianluigi Greco}, {and} \bibinfo{person}{Francesco
  Scarcello}.} \bibinfo{year}{2007}\natexlab{}.
\newblock \showarticletitle{Hypertree Decompositions for Query Optimization}.
  In \bibinfo{booktitle}{\emph{Proc.\ {ICDE} 2007}}. \bibinfo{publisher}{{IEEE}
  Computer Society}, \bibinfo{pages}{36--45}.
\newblock
\urldef\tempurl%
\url{https://doi.org/10.1109/ICDE.2007.367849}
\showDOI{\tempurl}


\bibitem[\protect\citeauthoryear{Ghionna, Greco, and Scarcello}{Ghionna
  et~al\mbox{.}}{2011}]%
        {DBLP:conf/cikm/GhionnaGS11}
\bibfield{author}{\bibinfo{person}{Lucantonio Ghionna},
  \bibinfo{person}{Gianluigi Greco}, {and} \bibinfo{person}{Francesco
  Scarcello}.} \bibinfo{year}{2011}\natexlab{}.
\newblock \showarticletitle{{H-DB:} a hybrid quantitative-structural sql
  optimizer}. In \bibinfo{booktitle}{\emph{Proc.\ {CIKM} 2011}}.
  \bibinfo{publisher}{{ACM}}, \bibinfo{pages}{2573--2576}.
\newblock


\bibitem[\protect\citeauthoryear{Gottlob, Greco, Leone, and Scarcello}{Gottlob
  et~al\mbox{.}}{2016}]%
        {DBLP:conf/pods/GottlobGLS16}
\bibfield{author}{\bibinfo{person}{Georg Gottlob}, \bibinfo{person}{Gianluigi
  Greco}, \bibinfo{person}{Nicola Leone}, {and} \bibinfo{person}{Francesco
  Scarcello}.} \bibinfo{year}{2016}\natexlab{}.
\newblock \showarticletitle{Hypertree Decompositions: Questions and Answers}.
  In \bibinfo{booktitle}{\emph{Proc.\ {PODS} 2016}}.
  \bibinfo{publisher}{{ACM}}, \bibinfo{pages}{57--74}.
\newblock


\bibitem[\protect\citeauthoryear{Gottlob, Leone, and Scarcello}{Gottlob
  et~al\mbox{.}}{2002}]%
        {2002gottlob}
\bibfield{author}{\bibinfo{person}{Georg Gottlob}, \bibinfo{person}{Nicola
  Leone}, {and} \bibinfo{person}{Francesco Scarcello}.}
  \bibinfo{year}{2002}\natexlab{}.
\newblock \showarticletitle{Hypertree Decompositions and Tractable Queries}.
\newblock \bibinfo{journal}{\emph{J. Comput. Syst. Sci.}} \bibinfo{volume}{64},
  \bibinfo{number}{3} (\bibinfo{year}{2002}), \bibinfo{pages}{579--627}.
\newblock


\bibitem[\protect\citeauthoryear{Gottlob, Mikl\'{o}s, and Schwentick}{Gottlob
  et~al\mbox{.}}{2009}]%
        {2009gottlob}
\bibfield{author}{\bibinfo{person}{Georg Gottlob}, \bibinfo{person}{Zolt\'{a}n
  Mikl\'{o}s}, {and} \bibinfo{person}{Thomas Schwentick}.}
  \bibinfo{year}{2009}\natexlab{}.
\newblock \showarticletitle{Generalized Hypertree Decompositions: {NP}-hardness
  and Tractable Variants}.
\newblock \bibinfo{journal}{\emph{J. ACM}} \bibinfo{volume}{56},
  \bibinfo{number}{6} (\bibinfo{year}{2009}), \bibinfo{pages}{30:1--30:32}.
\newblock


\bibitem[\protect\citeauthoryear{Gottlob and Samer}{Gottlob and Samer}{2008}]%
        {DBLP:journals/jea/GottlobS08}
\bibfield{author}{\bibinfo{person}{Georg Gottlob} {and} \bibinfo{person}{Marko
  Samer}.} \bibinfo{year}{2008}\natexlab{}.
\newblock \showarticletitle{A backtracking-based algorithm for hypertree
  decomposition}.
\newblock \bibinfo{journal}{\emph{{ACM} Journal of Experimental Algorithmics}}
  \bibinfo{volume}{13} (\bibinfo{year}{2008}).
\newblock


\bibitem[\protect\citeauthoryear{Grohe and Marx}{Grohe and Marx}{2014}]%
        {2014grohemarx}
\bibfield{author}{\bibinfo{person}{Martin Grohe} {and}
  \bibinfo{person}{D{\'{a}}niel Marx}.} \bibinfo{year}{2014}\natexlab{}.
\newblock \showarticletitle{Constraint Solving via Fractional Edge Covers}.
\newblock \bibinfo{journal}{\emph{{ACM} Trans. Algorithms}}
  \bibinfo{volume}{11}, \bibinfo{number}{1} (\bibinfo{year}{2014}),
  \bibinfo{pages}{4:1--4:20}.
\newblock


\bibitem[\protect\citeauthoryear{Guo, Pan, and Heflin}{Guo
  et~al\mbox{.}}{2005}]%
        {guo2005lubm}
\bibfield{author}{\bibinfo{person}{Yuanbo Guo}, \bibinfo{person}{Zhengxiang
  Pan}, {and} \bibinfo{person}{Jeff Heflin}.} \bibinfo{year}{2005}\natexlab{}.
\newblock \showarticletitle{{LUBM:} {A} benchmark for {OWL} knowledge base
  systems}.
\newblock \bibinfo{journal}{\emph{J. Web Sem.}} \bibinfo{volume}{3},
  \bibinfo{number}{2-3} (\bibinfo{year}{2005}), \bibinfo{pages}{158--182}.
\newblock
\urldef\tempurl%
\url{https://doi.org/10.1016/j.websem.2005.06.005}
\showDOI{\tempurl}


\bibitem[\protect\citeauthoryear{Habbas, Amroun, and Singer}{Habbas
  et~al\mbox{.}}{2015}]%
        {DBLP:journals/jetai/HabbasAS15}
\bibfield{author}{\bibinfo{person}{Zineb Habbas}, \bibinfo{person}{Kamal
  Amroun}, {and} \bibinfo{person}{Daniel Singer}.}
  \bibinfo{year}{2015}\natexlab{}.
\newblock \showarticletitle{A Forward-Checking algorithm based on a Generalised
  Hypertree Decomposition for solving non-binary constraint satisfaction
  problems}.
\newblock \bibinfo{journal}{\emph{J. Exp. Theor. Artif. Intell.}}
  \bibinfo{volume}{27}, \bibinfo{number}{5} (\bibinfo{year}{2015}),
  \bibinfo{pages}{649--671}.
\newblock
\urldef\tempurl%
\url{https://doi.org/10.1080/0952813X.2014.993507}
\showDOI{\tempurl}


\bibitem[\protect\citeauthoryear{Jain, Moritz, Halperin, Howe, and
  Lazowska}{Jain et~al\mbox{.}}{2016}]%
        {shrjainSQLShare}
\bibfield{author}{\bibinfo{person}{Shrainik Jain}, \bibinfo{person}{Dominik
  Moritz}, \bibinfo{person}{Daniel Halperin}, \bibinfo{person}{Bill Howe},
  {and} \bibinfo{person}{Ed Lazowska}.} \bibinfo{year}{2016}\natexlab{}.
\newblock \showarticletitle{SQLShare: Results from a Multi-Year
  SQL-as-a-Service Experiment}. In \bibinfo{booktitle}{\emph{Proceedings of the
  2016 International Conference on Management of Data}}
  \emph{(\bibinfo{series}{SIGMOD '16})}. \bibinfo{publisher}{ACM},
  \bibinfo{address}{New York, NY, USA}, \bibinfo{pages}{281--293}.
\newblock
\showISBNx{978-1-4503-3531-7}
\urldef\tempurl%
\url{https://doi.org/10.1145/2882903.2882957}
\showDOI{\tempurl}


\bibitem[\protect\citeauthoryear{Karakashian, Woodward, and
  Choueiry}{Karakashian et~al\mbox{.}}{2011}]%
        {DBLP:conf/sara/KarakashianWC11}
\bibfield{author}{\bibinfo{person}{Shant Karakashian},
  \bibinfo{person}{Robert~J. Woodward}, {and} \bibinfo{person}{Berthe~Y.
  Choueiry}.} \bibinfo{year}{2011}\natexlab{}.
\newblock \showarticletitle{Reformulating R(*, m)C with Tree Decomposition}. In
  \bibinfo{booktitle}{\emph{{SARA}}}. \bibinfo{publisher}{{AAAI}}.
\newblock


\bibitem[\protect\citeauthoryear{Khamis, Ngo, R{\'{e}}, and Rudra}{Khamis
  et~al\mbox{.}}{2015}]%
        {KhamisNRR15}
\bibfield{author}{\bibinfo{person}{Mahmoud~Abo Khamis},
  \bibinfo{person}{Hung~Q. Ngo}, \bibinfo{person}{Christopher R{\'{e}}}, {and}
  \bibinfo{person}{Atri Rudra}.} \bibinfo{year}{2015}\natexlab{}.
\newblock \showarticletitle{Joins via Geometric Resolutions: Worst-case and
  Beyond}. In \bibinfo{booktitle}{\emph{Proc.\ {PODS} 2015}}.
\newblock


\bibitem[\protect\citeauthoryear{Khamis, Ngo, and Rudra}{Khamis
  et~al\mbox{.}}{2016}]%
        {KhamisNR16}
\bibfield{author}{\bibinfo{person}{Mahmoud~Abo Khamis},
  \bibinfo{person}{Hung~Q. Ngo}, {and} \bibinfo{person}{Atri Rudra}.}
  \bibinfo{year}{2016}\natexlab{}.
\newblock \showarticletitle{{FAQ:} Questions Asked Frequently}. In
  \bibinfo{booktitle}{\emph{Proc.\ {PODS} 2016}}. \bibinfo{pages}{13--28}.
\newblock


\bibitem[\protect\citeauthoryear{Lalou, Habbas, and Amroun}{Lalou
  et~al\mbox{.}}{2009}]%
        {LalouHA09}
\bibfield{author}{\bibinfo{person}{Mohammed Lalou}, \bibinfo{person}{Zineb
  Habbas}, {and} \bibinfo{person}{Kamal Amroun}.}
  \bibinfo{year}{2009}\natexlab{}.
\newblock \showarticletitle{Solving Hypertree Structured {CSP:} Sequential and
  Parallel Approaches}. In \bibinfo{booktitle}{\emph{Proc.\ RCRA@AI*IA 2009}}.
\newblock


\bibitem[\protect\citeauthoryear{Leis, Gubichev, Mirchev, Boncz, Kemper, and
  Neumann}{Leis et~al\mbox{.}}{2015}]%
        {2015leis}
\bibfield{author}{\bibinfo{person}{Viktor Leis}, \bibinfo{person}{Andrey
  Gubichev}, \bibinfo{person}{Atanas Mirchev}, \bibinfo{person}{Peter Boncz},
  \bibinfo{person}{Alfons Kemper}, {and} \bibinfo{person}{Thomas Neumann}.}
  \bibinfo{year}{2015}\natexlab{}.
\newblock \showarticletitle{How Good Are Query Optimizers, Really?}
\newblock \bibinfo{journal}{\emph{PVLDB}} \bibinfo{volume}{9},
  \bibinfo{number}{3} (\bibinfo{date}{Nov.} \bibinfo{year}{2015}),
  \bibinfo{pages}{204--215}.
\newblock
\showISSN{2150-8097}
\urldef\tempurl%
\url{https://doi.org/10.14778/2850583.2850594}
\showDOI{\tempurl}


\bibitem[\protect\citeauthoryear{Leis, Radke, Gubichev, Mirchev, Boncz, Kemper,
  and Neumann}{Leis et~al\mbox{.}}{2017}]%
        {Leis2017}
\bibfield{author}{\bibinfo{person}{Viktor Leis}, \bibinfo{person}{Bernhard
  Radke}, \bibinfo{person}{Andrey Gubichev}, \bibinfo{person}{Atanas Mirchev},
  \bibinfo{person}{Peter Boncz}, \bibinfo{person}{Alfons Kemper}, {and}
  \bibinfo{person}{Thomas Neumann}.} \bibinfo{year}{2017}\natexlab{}.
\newblock \showarticletitle{Query optimization through the looking glass, and
  what we found running the Join Order Benchmark}.
\newblock \bibinfo{journal}{\emph{The VLDB Journal}} (\bibinfo{date}{18 Sep}
  \bibinfo{year}{2017}).
\newblock
\showISSN{0949-877X}
\urldef\tempurl%
\url{https://doi.org/10.1007/s00778-017-0480-7}
\showDOI{\tempurl}


\bibitem[\protect\citeauthoryear{Malyshev, Kr{\"o}tzsch, Gonz{\'a}lez, Gonsior,
  and Bielefeldt}{Malyshev et~al\mbox{.}}{2018}]%
        {malyshevgetting}
\bibfield{author}{\bibinfo{person}{Stanislav Malyshev}, \bibinfo{person}{Markus
  Kr{\"o}tzsch}, \bibinfo{person}{Larry Gonz{\'a}lez}, \bibinfo{person}{Julius
  Gonsior}, {and} \bibinfo{person}{Adrian Bielefeldt}.}
  \bibinfo{year}{2018}\natexlab{}.
\newblock \showarticletitle{Getting the Most out of Wikidata: Semantic
  Technology Usage in Wikipedia’s Knowledge Graph}. In
  \bibinfo{booktitle}{\emph{Proc.\ ISWC 2018}}.
\newblock
\newblock
\shownote{To appear.}


\bibitem[\protect\citeauthoryear{Marx}{Marx}{2010}]%
        {DBLP:journals/talg/Marx10}
\bibfield{author}{\bibinfo{person}{D\'{a}niel Marx}.}
  \bibinfo{year}{2010}\natexlab{}.
\newblock \showarticletitle{Approximating Fractional Hypertree Width}.
\newblock \bibinfo{journal}{\emph{ACM Trans. Algorithms}} \bibinfo{volume}{6},
  \bibinfo{number}{2}, Article \bibinfo{articleno}{29} (\bibinfo{year}{2010}),
  \bibinfo{numpages}{29:1--29:17}~pages.
\newblock


\bibitem[\protect\citeauthoryear{Moll, Tazari, and Thurley}{Moll
  et~al\mbox{.}}{2012}]%
        {moll2012}
\bibfield{author}{\bibinfo{person}{Lukas Moll}, \bibinfo{person}{Siamak
  Tazari}, {and} \bibinfo{person}{Marc Thurley}.}
  \bibinfo{year}{2012}\natexlab{}.
\newblock \showarticletitle{Computing hypergraph width measures exactly}.
\newblock \bibinfo{journal}{\emph{Inf. Process. Lett.}} \bibinfo{volume}{112},
  \bibinfo{number}{6} (\bibinfo{year}{2012}), \bibinfo{pages}{238--242}.
\newblock


\bibitem[\protect\citeauthoryear{Olteanu and Z{\'a}vodn{\`y}}{Olteanu and
  Z{\'a}vodn{\`y}}{2015}]%
        {olteanu2015size}
\bibfield{author}{\bibinfo{person}{Dan Olteanu} {and} \bibinfo{person}{Jakub
  Z{\'a}vodn{\`y}}.} \bibinfo{year}{2015}\natexlab{}.
\newblock \showarticletitle{Size bounds for factorised representations of query
  results}.
\newblock \bibinfo{journal}{\emph{{ACM} Trans. Database Syst.}}
  \bibinfo{volume}{40}, \bibinfo{number}{1} (\bibinfo{year}{2015}),
  \bibinfo{pages}{2}.
\newblock


\bibitem[\protect\citeauthoryear{Perelman and R{\'e}}{Perelman and
  R{\'e}}{2015}]%
        {Duncecap15a}
\bibfield{author}{\bibinfo{person}{Adam Perelman} {and}
  \bibinfo{person}{Christopher R{\'e}}.} \bibinfo{year}{2015}\natexlab{}.
\newblock \showarticletitle{{DunceCap}: Compiling Worst-Case Optimal Query
  Plans}. In \bibinfo{booktitle}{\emph{Proc.\ {SIGMOD} 2015}}.
  \bibinfo{publisher}{{ACM}}, \bibinfo{pages}{2075--2076}.
\newblock


\bibitem[\protect\citeauthoryear{Picalausa and Vansummeren}{Picalausa and
  Vansummeren}{2011}]%
        {DBLP:conf/sigmod/PicalausaV11}
\bibfield{author}{\bibinfo{person}{Fran{\c{c}}ois Picalausa} {and}
  \bibinfo{person}{Stijn Vansummeren}.} \bibinfo{year}{2011}\natexlab{}.
\newblock \showarticletitle{What are real {SPARQL} queries like?}. In
  \bibinfo{booktitle}{\emph{Proc.\ {SWIM} 2011}}. \bibinfo{publisher}{{ACM}},
  \bibinfo{pages}{7}.
\newblock
\urldef\tempurl%
\url{https://doi.org/10.1145/1999299.1999306}
\showDOI{\tempurl}


\bibitem[\protect\citeauthoryear{Pottinger and Halevy}{Pottinger and
  Halevy}{2001}]%
        {2001Pottinger}
\bibfield{author}{\bibinfo{person}{Rachel Pottinger} {and}
  \bibinfo{person}{Alon Halevy}.} \bibinfo{year}{2001}\natexlab{}.
\newblock \showarticletitle{{MiniCon}: A Scalable Algorithm for Answering
  Queries Using Views}.
\newblock \bibinfo{journal}{\emph{The VLDB Journal}} \bibinfo{volume}{10},
  \bibinfo{number}{2-3} (\bibinfo{date}{Sept.} \bibinfo{year}{2001}),
  \bibinfo{pages}{182--198}.
\newblock
\showISSN{1066-8888}


\bibitem[\protect\citeauthoryear{Scarcello, Greco, and Leone}{Scarcello
  et~al\mbox{.}}{2007}]%
        {DBLP:journals/jcss/ScarcelloGL07}
\bibfield{author}{\bibinfo{person}{Francesco Scarcello},
  \bibinfo{person}{Gianluigi Greco}, {and} \bibinfo{person}{Nicola Leone}.}
  \bibinfo{year}{2007}\natexlab{}.
\newblock \showarticletitle{Weighted hypertree decompositions and optimal query
  plans}.
\newblock \bibinfo{journal}{\emph{J. Comput. Syst. Sci.}} \bibinfo{volume}{73},
  \bibinfo{number}{3} (\bibinfo{year}{2007}).
\newblock


\bibitem[\protect\citeauthoryear{Schafhauser}{Schafhauser}{2006}]%
        {Schafhauser06}
\bibfield{author}{\bibinfo{person}{Werner Schafhauser}.}
  \bibinfo{year}{2006}\natexlab{}.
\newblock \bibinfo{title}{New heuristic methods for tree decompositions and
  generalized hypertree decompositions}.
\newblock   (\bibinfo{year}{2006}).
\newblock
\newblock
\shownote{Master Thesis, TU Wien.}


\bibitem[\protect\citeauthoryear{Schild and Sommer}{Schild and Sommer}{2015}]%
        {DBLP:conf/wea/SchildS15}
\bibfield{author}{\bibinfo{person}{Aaron Schild} {and}
  \bibinfo{person}{Christian Sommer}.} \bibinfo{year}{2015}\natexlab{}.
\newblock \showarticletitle{On Balanced Separators in Road Networks}. In
  \bibinfo{booktitle}{\emph{Proc.\ {SEA} 2015}}
  \emph{(\bibinfo{series}{LNCS})}, Vol.~\bibinfo{volume}{9125}.
  \bibinfo{publisher}{Springer}, \bibinfo{pages}{286--297}.
\newblock


\bibitem[\protect\citeauthoryear{{Transaction Processing Performance Council
  (TPC)}}{{Transaction Processing Performance Council (TPC)}}{2014}]%
        {tpch}
\bibfield{author}{\bibinfo{person}{{Transaction Processing Performance Council
  (TPC)}}.} \bibinfo{year}{2014}\natexlab{}.
\newblock \bibinfo{title}{{TPC-H} decision support benchmark}.
\newblock \bibinfo{howpublished}{\url{http://www.tpc.org/tpch/default.asp}}.
  (\bibinfo{year}{2014}).
\newblock


\bibitem[\protect\citeauthoryear{Tu and R{\'e}}{Tu and R{\'e}}{2015}]%
        {tu2015duncecap}
\bibfield{author}{\bibinfo{person}{Susan Tu} {and} \bibinfo{person}{Christopher
  R{\'e}}.} \bibinfo{year}{2015}\natexlab{}.
\newblock \showarticletitle{Duncecap: Query plans using generalized hypertree
  decompositions}. In \bibinfo{booktitle}{\emph{Proc.\ {SIGMOD} 2015}}. ACM,
  \bibinfo{pages}{2077--2078}.
\newblock


\bibitem[\protect\citeauthoryear{Vapnik and Chervonenkis}{Vapnik and
  Chervonenkis}{1971}]%
        {1971vc}
\bibfield{author}{\bibinfo{person}{Vladimir Vapnik} {and}
  \bibinfo{person}{Alexey Chervonenkis}.} \bibinfo{year}{1971}\natexlab{}.
\newblock \showarticletitle{On the uniform convergence of relative frequencies
  of events to their probabilities}.
\newblock \bibinfo{journal}{\emph{Theory Probab. Appl.}}  \bibinfo{volume}{16}
  (\bibinfo{year}{1971}), \bibinfo{pages}{264--280}.
\newblock


\end{thebibliography}
\end{document}